\pdfoutput=1

\documentclass{article}

\usepackage[pdftex]{graphicx}
\graphicspath{{Figures/}}

\usepackage{latexsym}
\usepackage{amssymb}
\usepackage[cmex10]{amsmath}
\usepackage{amsthm}
\newtheorem{theorem}{Theorem}

\usepackage{algorithmic}

\usepackage{array}
\usepackage{url}

\begin{document}
\title{Nonadaptive Mastermind Algorithms for String and Vector Databases, with Case Studies}

\author{Arthur U. Asuncion and Michael T. Goodrich\footnote{A. Asuncion and M. Goodrich are with the Department
of Computer Science, University of California, Irvine, CA, 92697.
E-mail: \{asuncion, goodrich\}@ics.uci.edu}}

\maketitle

\begin{abstract}
In this paper, we study sparsity-exploiting Mastermind algorithms for attacking 
the privacy of an entire database of character strings or vectors, such 
as DNA strings, movie ratings, or social network friendship data.  Based on reductions to 
nonadaptive group testing, our methods are able to take advantage of minimal 
amounts of privacy leakage, such as contained in a single bit that indicates if 
two people in a medical database have any common genetic mutations, or if two 
people have any common friends in an online social network.  We analyze our 
Mastermind attack algorithms using theoretical characterizations that provide
sublinear bounds on the number of queries needed to clone the database, as well as 
experimental tests on genomic information, collaborative filtering data, 
and online social networks.  By taking advantage of the generally sparse nature 
of these real-world databases and modulating a parameter that controls query sparsity, 
we demonstrate that relatively few nonadaptive queries are needed to recover a 
large majority of each database.
\end{abstract}

\section{Introduction}
Privacy and data protection are important and growing concerns 
when dealing with character strings or vector data. Medical databases are 
constrained by Health Insurance Portability and Accountability Act (HIPAA) rules 
to keep identifying data private, for instance. 
Such databases in the future will commonly store DNA strings of patients,
which will need to have their privacy protected for obvious reasons.
Likewise, attribute vectors, which reflect the presence or absence of
each of a large number of possible attributes, are common in biotechnology; 
for example, chemical attribute vectors (e.g., see~\cite{baldicompression07,swamidasssearch07})
indicate the presence or absence of each of about a million attributes.

Privacy concerns also exist for online social networks and other databases 
which store user preferences in vector form.  For instance, knowledge of a 
social network user's set of friends (representable as a row in an adjacency matrix)
is potentially a gateway privacy leak, for friendship overlaps have been shown to  
be sufficient to de-anonymize individuals across multiple social
networking sites~\cite{nv-dsn-09}.  Likewise, the movie rating vectors in the 
database used for the Netflix Prize contest consists of ratings of movies by 
individual users, which are generally deemed as sensitive information.
Full access to such databases may be constrained by privacy
agreements or legitimate proprietary reasons for keeping these databases
private, even as they allow for limited types of queries to be performed on them.

Each time a client queries such a database and it responds 
with an answer, it reveals some information about
its contents, even if the client and the database are using
a Secure Multiparty Computation (SMC) protocol (e.g., 
see~\cite{akd-spsc-03,bnp-fssmp-08,dfknt-uscfm-06,jks-tppgc-08,
jmcs-sddli-08,spol-tpdp-06,ss-ppsip-07,y-psc-82})
to process such a query.
Thus, we can provide a crude characterization of 
the risk of privacy loss in biological, medical, or proprietary databases
in terms of the existence of efficient algorithms that can take advantage 
of the data leakage present in query responses to
be able to replicate part or all of the content of the database.
We refer to such schemes as \emph{data-cloning} attacks.

Formally, in an \emph{algorithmic data-cloning attack}, a \emph{querier}, Bob, is
allowed certain types of queries to a database, ${\cal X}$, that belongs
to a \emph{data owner}, Alice. Bob's goal is to replicate all or a large part 
of ${\cal X}$ through as few  
queries on ${\cal X}$ as possible (and with low computational overhead).
In this paper, we focus on databases where ${\cal X}$ is a
collection ${\cal X}=(X_1,X_2,\ldots,X_g)$ of character strings or vectors,
over a fixed-size alphabet.  With respect to the types of databases we consider,
we assume that Alice is willing to process \emph{comparison}
queries from Bob, each of which consists of 
Bob providing a single vector $Q$ (which is not
necessarily revealed in plaintext to Alice) and, 
possibly using a Secure Multiparty Computation (SMC).
Alice reveals a response vector $(r_1,r_2,\ldots,r_g)$, where 
each $r_i$ is the score for some type of comparison of $Q$ with $X_i$.
In the simplest case, each score $r_i$ can be a single bit denoting 
whether the query $Q$ shares any common entries with $X_i$.
As mentioned above, the risk to this \emph{data-cloning} attack, then,
can be characterized by the number of queries 
and how much processing time is needed so that 
Bob can replicate all of ${\cal X}$ or a large portion of ${\cal X}$.

\subsection{Our Contributions}
Inspired by a game known as \emph{Mastermind}, we present a number of 
algorithms for performing a \emph{Mastermind attack} on an entire string 
or vector database, ${\cal X}=(X_1,X_2,\ldots,X_g)$, so as to clone all or
a large portion of ${\cal X}$.  All of our methods assume only the SMC 
protocol of Jiang {\it et al.}~\cite{jmcs-sddli-08}, where a querier, Bob,
issues a query string or vector, $Q$, and receives a vector of
responses $(r_1,r_2,\ldots,r_g)$, where each $r_i$ is a single numerical response 
score measuring the similarity of $Q$ and $X_i$ according to some
public metric.  Since vectors taken over a universe of size $c$ can be viewed as
character strings taken over an alphabet of size $c$, we will,
without loss of generality, focus our descriptions on the case when
${\cal X}$ consists of $g$ character strings.
We will also assume that each string in ${\cal X}$ is the same length, since
we can view smaller strings as being padded with an additional
character not in the original alphabet.

We show that repeated querying of such a database can clone all or a large
portion of it, often with a surprisingly small, \emph{sublinear} number of queries.
The risk profile we explore in each type of attack, then, is the
number of queries needed to execute it.
Specifically, let us suppose that ${\cal X}$ contains $g$ strings, each of length $n$, 
taken over an alphabet of size $c$,
with at least $g'\le g$ of these strings having at most $d<n$
differences from a public reference string, $R$.  We show that 
at least $g'$ of the strings in ${\cal X}$ can be cloned using
(at most) the following number of queries:
\[
 2(c-1)(2d\log n + \min\{d\log g, d^2\log (en/d)\}).
\]

This result applies to situations, common in
many real-world databases (e.g. ~\cite{baldicompression07,baldidcc08,swamidasssearch07}), 
where strings in the database
can be characterized in terms of a small number of differences with a
reference string, $R$.

We also provide several case studies showing empirical data that 
demonstrates that our
randomized attack can work effectively on real-world
databases. For instance, 
we apply our attack to a database of mitochondrial DNA (mtDNA) strings
and the database of movie-ratings vectors provided for the Netflix
Prize contest, showing that large portions of these databases can be
cloned using a number of queries that is much smaller than the length
of the strings or vectors in these databases.

If, in practice, Bob learns more than the information contained in the 
response vector $(r_1,r_2,\ldots,r_g)$, that only strengthens his attack. The
point of this paper is that even with just the information leaked in
the responses, Bob can construct a small number of query vectors that are sufficient to 
learn all or a sizeable fraction of the vectors in $\cal X$.
Moreover, our Mastermind attack is \emph{oblivious} (that is, \emph{nonadaptive}), 
in that Bob can construct all his query vectors in advance, so that
the format of no query depends on the outcome of another. 
We describe a randomized construction for Bob's query vectors, which
allows the attack to be fairly surreptitious, in that each query 
looks random (because it is random).

\section{Attack Scenarios}
Before describing our nonadaptive Mastermind attack in precise detail, 
we show how it applies to a wide variety of attack scenarios to provide
motivating examples.  We illustrate three such attack scenarios
below.

\subsection{Genetic Signatures}

\begin{figure}[]
\begin{center}
\includegraphics[width=3.2in]{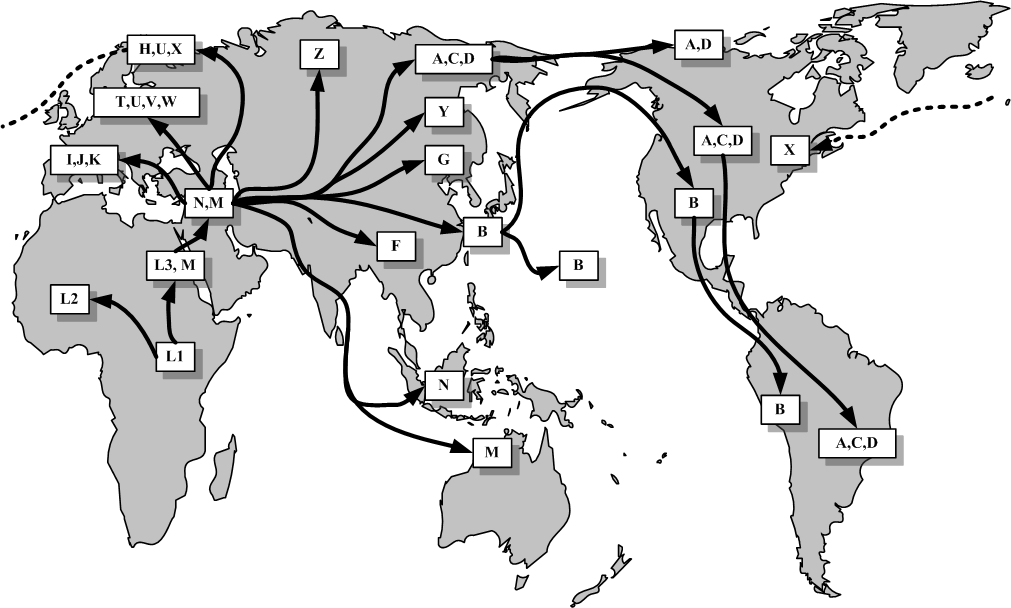}
\end{center}
\vspace*{-12pt}
\caption{
An illustration of
the pattern of human migration together with major mutations in
human mtDNA~\cite{ps-mdhe-05,brb-gppd-05}, which is only transfered along
the maternal line.
Each letter stands for a set of mutations from the reference
string, $R$.
Thus, determining locations of
differences with $R$ can reveal ethnic identity, sometimes to the
resolution of the village of maternal ancestry.
(Image, Copyright 2009, Michael T. Goodrich. Used with permission.)
}
\label{fig:mtDNA}
\end{figure}

Suppose the vectors in $\cal X$ represent the genetic signatures of people
in some population, such as a high school, college, or corporation.
Bob's goal in this Mastermind attack is to learn the genetic
signatures for as many people in his population of interest as
is reasonably possible.
He can employ his attack so long as there is a website or tool for
$\cal X$ that allows him to test a query vector $Q$ against the
vectors in $\cal X$ to determine which ones share a mutation with
$Q$, with respect to a reference $R$.  In mitochondrial DNA, 
the reference $R$ is roughly 16,500 base pairs long, 
but has only about 4,000 known mutations~\cite{brandon04,ruiz07}, 
suggesting that each vector in $\cal X$ is sparse relative to $R$.

In this example, Bob could be posing as a medical researcher and claim that 
his vectors are testing for combinations of genetic markers for disease. 
Alternatively he could claim to be a forensic analyst with DNA from a 
crime scene, which he wants to test against members of $\cal X$ 
(in this case, he is likely to receive a similarity score between his 
query $Q$ and the vectors in $\cal X$, which he can easily convert 
into an overlap-detection bit). In either case, a minimum amount of 
overlap information can allow him to learn the entire genetic signatures 
of a large number of members of $\cal X$.

The privacy implications of such an innocuous attack are significant.  
Alice's genetic signature could then be used by an unethical employer or 
insurance company to discriminate against her based on his risks for future diseases.
Also, as illustrated in Figure~\ref{fig:mtDNA}, it is possible using a genetic signature derived from
a short string of Alice's mitochondrial DNA to trace her 
maternal lineage to an 
ancestral location~\cite{brb-gppd-05,ps-mdhe-05}, which is
information that could then be used for ethnic 
discrimination~\cite{harihara92}.

\subsection{Social Network Friendship Ties}

\begin{figure}[]
\begin{minipage}[b]{0.5\linewidth}\centering
\includegraphics[height=1.7in]{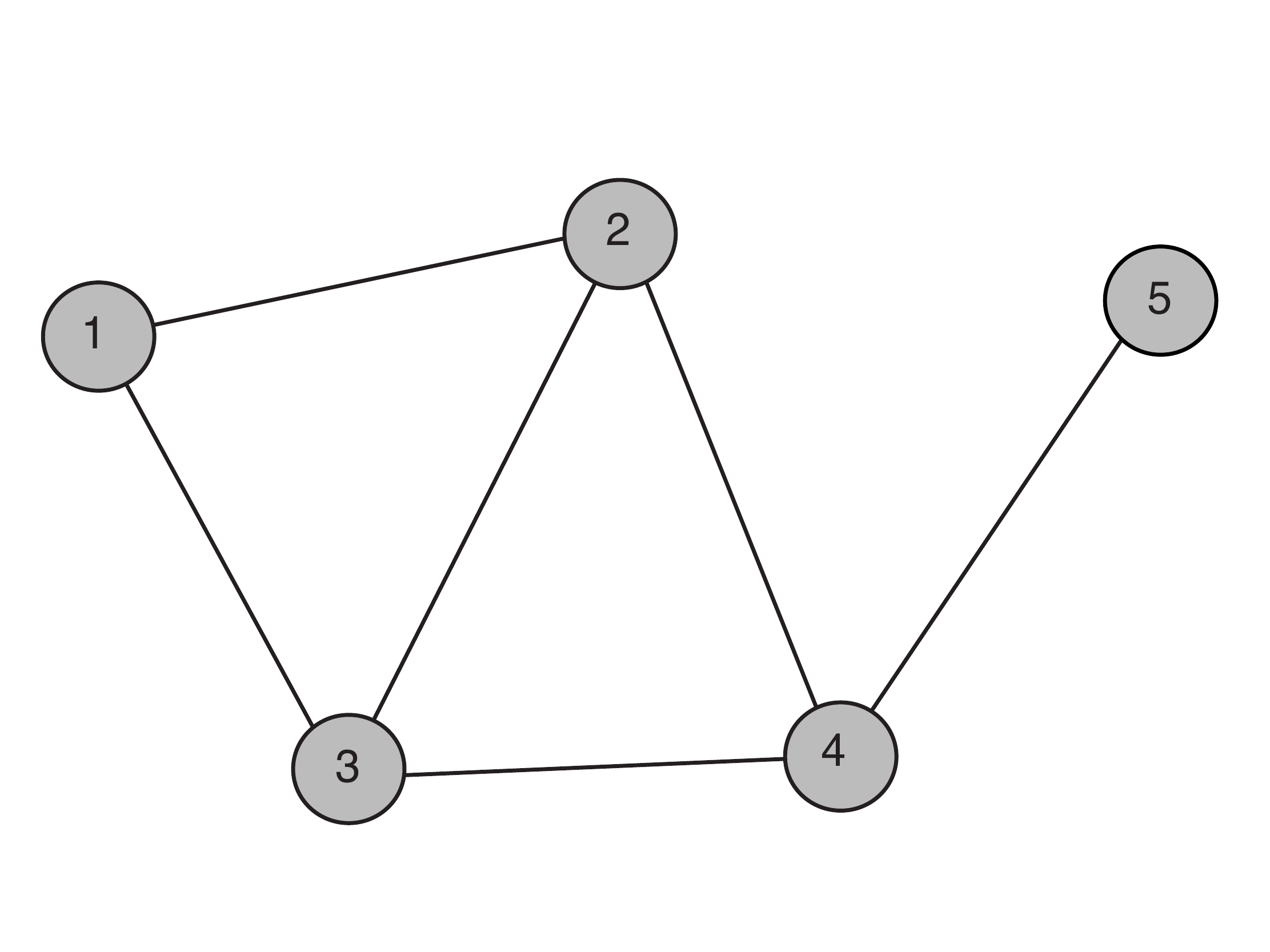}
\end{minipage}
\begin{minipage}[b]{0.5\linewidth}
\begin{tabular}{c|ccccc}
 & 1 & 2 & 3 & 4 & 5 \\
\hline
1 & 0 & 1 & 1 & 0 & 0 \\
2 & 1 & 0 & 1 & 1 & 0 \\
3 & 1 & 1 & 0 & 1 & 0 \\
4 & 0 & 1 & 1 & 0 & 1 \\
5 & 0 & 0 & 0 & 1 & 0 
\end{tabular}
\end{minipage}
\caption{An example graph and its adjacency matrix.}
\label{fig:matrix}
\end{figure}

Suppose the vectors in $\cal X$ represent the rows of the 
adjacency matrix (e.g., Figure~\ref{fig:matrix}) defined by the friendship ties for an online social network, 
like Facebook, possibly restricted to the population in a
specific city, college, high school, or large corporation.
In this scenario, Bob wants to learn the friendship relationships of
as many people as possible. For instance, he may wish to do racial
profiling~\cite{Lewis2008330} or do a cross-networking
identification attack~\cite{nv-dsn-09}, since 
89\% of Facebook users use their real names~\cite{gah-irpos-05}.

In this case, Bob's query vectors correspond to a relatively small
number of pseudonyms that Bob creates in the
social network and for which he defines a certain number of random friendship ties.
For instance, he could create such ties using automated social
engineering techniques (e.g., using the name of an affiliated city,
college, etc.) as well as the property that 
a fairly large percentage of social networking users are likely to accept
random friendship requests from people in their community 
(roughly 10 to 25\% of student Facebook 
users accept random friendship requests from people who say they are
in the same university~\cite{st-snf-07}).
Given his set of pseudonyms, Bob employs the group-testing
attack by having each of his pseudonyms ask the social networking site 
if this pseudonym shares any friends with the people in Bob's
population of interest.
Note that he will receive a useful response vector from everyone that
has privacy settings that allow for testing for mutual friends in
common.
That is, even if someone chooses to share friendship
information only with ``friends of friends,'' which is one of the more
restrictive standard privacy settings in Facebook, Bob can still get
valid responses for his queries with respect to such people.
Moreover, if Bob employs an oblivious group-testing attack, he can
use the same set of pseudonyms for everyone whose privacy he is
attacking.
Thus, once he has set up his pseudonyms, he can target 
the privacy of any user in
the online social network at will.

\subsection{User Preference Data}
Suppose the vectors in $\cal X$ represent the preferences of people
in a site, such as Amazon or Netflix,
that employs collaborative filtering to support product recommendations.
Specifically, we assume in this scenario that products are numbered
$1$ to $k$ and each vector $X_i$ in $\cal X$ has a discrete rating (e.g. 1-5 stars, 
or a missing rating) in position $j$, provided by user $i$.
Bob's goal in this scenario is to discover as many vectors in $\cal X$ as 
reasonably possible and in 
so doing discover the product preferences of a large number of
targeted people.
His motivation could, for instance, be economic, in that he may want
to open an online store that caters to
a specific demographic; hence, we may want to learn the product preferences for
a known population of people in this group.
In terms of information leakage, all that is needed in order to allow
for Bob's group-testing attack to work is for the collaborative
filtering site have a way for him to create pseudonyms, have these
pseudonyms rate products, and allow for these pseudonyms to test if 
they share any ratings in common with users in the target
population.
So long as the collaborative filtering web site allows for users to
check for overlapping scores with other users, Bob can employ the
nonadaptive Mastermind attack.

\subsection{Exploiting Sparsity}
The above set of attack
scenarios are illustrative of the risks to privacy that the
group-testing attack provides, in that it can greatly amplify 
the information gained from just a relatively small number
of single-bit privacy leaks.
The risk to the group-testing attack can
be characterized in terms of the number of 
queries and how much processing time is needed so that 
Bob can replicate a large portion of ${\cal X}$.
As we will elaborate in Section~\ref{sec:sparsity}, the critical 
factor here is a sparsity parameter, $d$, which, in a group testing
context, refers to the small number of ``defective'' items in the large
group.

Interestingly, each of the attack scenarios mentioned above
possess such a parameter,
allowing for Bob to employ efficient Mastermind attacks with a
relatively small number of queries.
For example, an individual's genetic signature will
typically have a relatively small number of indicators for mutations
with respect to a reference DNA string -- 
with mitochondrial DNA, most people have fewer than 100
mutations with respect to a commonly-used reference string.
Furthermore, most people in social networking sites, such as
Facebook, have less than a few hundred friends.
Likewise, most collaborative filtering preference vectors, such 
as in the Netflix Prize contest, have
ratings for at most a few hundred items.
Thus, there are several modern contexts 
that have all the pieces in place to allow for the Mastermind
attack to be used.

It is worth noting that realistic attacks can also be constructed in many other domains.  
For instance, sensitive image data, such as captured by biometric devices, may be 
represented as sparse vectors, making it susceptible to a 
Mastermind attack, especially when efficient tools exist for
comparing a query (e.g. a fingerprint or an iris scan) to the entire
database.

\section{Background and Related Work}

We give a brief background of the Mastermind game and attacks inspired by that game,
as well as related work on privacy models and attempts to mitigate privacy leaks.

\subsection{Mastermind}
Adapting the terminology of the Mastermind attack~\cite{g-magd-09} 
to attacks on an entire database, we discuss in this section the
relationship between the Mastermind attack and the Mastermind
boardgame.  Mastermind~\cite{c-m-83,k-cmm-77} 
is a two-player board game, 
which is played between a \emph{codemaker} and a
\emph{codebreaker}, using colored pegs
(Figure~\ref{fig-mastermind}).
Mastermind begins with
the codemaker selecting a character string, $X$, of length $n$, using an
alphabet of size $c$, whose members are called ``colors.''
The codebreaker then makes a sequence of queries, 
$Q_1,Q_2,\ldots$, about $X$'s identity.
For each guess $Q_i$, the codemaker provides a score on how well
$Q_i$ matches $X$.
In the board game, this is done using colored pegs, but we assume in this
paper that the score is simply a matching function,
$
b(Q_i)=|\{j\colon\, Q_i[j]=X[j]\}|,
$
which counts the number of places where $Q_i$ and $X$ match.
The codebreaker, of course, is trying to discover $X$ using 
a small number of guesses.

\begin{figure}
\begin{center}
\includegraphics[height=2.8in]{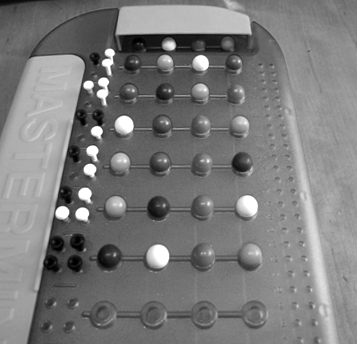}
\caption{The Mastermind game. The four large pegs in the middle are
used for guessing. The four smaller peg locations on the left are
used to score each guess---with black-peg and white-peg scores. 
(Image, Copyright 2009, Michael T. Goodrich. Used with permission.)
}
\label{fig-mastermind}
\end{center}
\end{figure}

Chv{\'a}tal~\cite{c-m-83} studied the combinatorics of 
the general Mastermind game, showing that it
can be solved in polynomial time using 
$
2n \lceil\log c\rceil + 4n
$
guesses. 
Chen {\it et al.}~\cite{cch-fhcaq-96} showed how this can be
improved to
$
2n \lceil\log n\rceil + 2n + \lceil c/n\rceil + 2
$
guesses and
Goodrich~\cite{g-oacmg-09} showed how this bound can be improved to
$ n\lceil \log c\rceil + \lceil(2-1/c)n\rceil + c$.
Unfortunately, from the perspective of the cloning problem,
all of these algorithms are \emph{adaptive}, in that they use
results of previous queries to construct future queries.
Adaptive algorithms can only be used effectively for the interactions between
a single pair of strings. For a sequence of queries to be used against an
entire database of strings, we need a \emph{nonadaptive} algorithm, 
that is, an algorithm where queries are not dependent upon answers
from previous queries, which is
equivalent to the codebreaker making all his guesses in advance.

Chv{\'a}tal~\cite{c-m-83} also gives an existence proof for a 
nonadaptive method for solving Mastermind.
If the number of possible colors, $c\le n^{1-\epsilon}$, 
for some constant $\epsilon>0$, which will almost
always be the case for biological databases, Chv{\'a}tal shows the existence
of a nonadaptive method using only 
\[
G=(2+\epsilon)n \frac{1+2\log c}{\log n - \log c}
\]
guesses.
In fact, he shows that making $G$ guesses at random will be sufficient to
determine a unique solution with high probability, using only the
$b(Q_i)$ type of scores.
Unfortunately, this existence proof does not immediately lead to a
polynomial-time algorithm.
Indeed, it is NP-complete to determine if a collection of Mastermind guesses
with $b(Q_i)$ type of responses is satisfiable~\cite{g-oacmg-09}.  Nonetheless,
in this paper, we will show that Mastermind attacks based on reductions to
group testing can efficiently clone a sparse database of interest
using a sublinear number of nonadaptive queries.

\subsection{Related Privacy Models}

Following a framework by 
Bancilhon and Spyratos~\cite{bs-pirdd-77},
Deutsch and Papakonstantinou~\cite{dp-pidp-05}
and Miklau and Suciu~\cite{ms-faidd-07}
give related models for
characterizing privacy loss in information releases from a database, 
which they call \emph{query-view security}.
In this framework, there is a
secret, $\cal S$, that the data owner, Alice, is trying to protect.
Attackers are allowed to ask legal queries of the database,
while Alice tries to protect the information that these
queries leak about $\cal S$.
While this framework is related to the data-cloning attack,
these two are not identical, since in the data-cloning attack
there is no specifically sensitive part of the
data. Instead, Alice, is trying to limit releasing too much of her data
to Bob rather than protecting any specific secret.
Similarly, 
Kantarcio\v{g}lu {\it et al.}~\cite{kjc-wdmvp-04} study privacy models that quantify
the degree to which data mining searches expose private information, but this
related privacy model is also not directly applicable to the data-cloning attack.

There has been considerable recent work on data modification
approaches that can help protect the privacy or
intellectual property rights of a database by modifying its content.
For example, several researchers
(e.g., see~\cite{ak-wrd-02,ahk-swrd-03,g-qpwrd-03,%
sv-hcwsd-04,sap-rprd-03,s-rard-07})
advocate the use of data watermarking to protect data rights.
In using this technique, data values are altered to make it easier,
after the fact, to track when someone has stolen information
from a database.
Of course, by that point, the data has already been cloned.
Alternatively, several
other researchers (e.g.,~\cite{lefevre05,samarati01,samarati98,mw-cok-04,%
afkmptz-at-05,bkbl-ekuct-07,zyw-pekcd-05,ba-dpoka-05}) propose
using \emph{generalization} or \emph{cell suppression}
as methods for 
achieving quantifiable privacy-preservation in databases.
These techniques alter data values to protect sensitive parts of the data,
while still allowing for data mining activities to be performed on the
database.
We assume here that Alice is not
interested in data modification techniques, however, for we believe that 
accuracy is critically important in several database applications.
For example, even a single
base-pair mutation in a DNA string can indicate the existence of an increased
health risk.

As mentioned above, we allow for the queries Bob asks to be answered
using SMC protocols, which reveal no additional information between the query
string $Q$ and each database string $X_i$ other than the response score $r_i$.
Such protocols have been developed for
the kinds of comparisons that are done in
genomic sequences (e.g., see~\cite{akd-spsc-03,da-smc-01,FNP04}).
In particular, Atallah {\it et al.}~\cite{akd-spsc-03} and
Atallah and Li~\cite{al-sosc-05}
studied
privacy-preserving protocols for edit-distance sequence comparisons, such as
in the longest common subsequence (LCS) 
problem (e.g.,~\cite{h-lsacm-75, ir-acvlc-08, uah-bclcs-76}).
Troncoso-Pastoriza {\it et al.}~\cite{tkc-pperd-07}
described a privacy-preserving protocol for regular-expression searching 
in a DNA sequence.
Jha {\it et al.}~\cite{jks-tppgc-08} give
privacy-preserving protocols for computing edit distance and
Smith-Waterman similarity scores between two genomic sequences, improving
the privacy-preserving algorithm of Szajda {\it et al.}~\cite{spol-tpdp-06}.
Aligned matching results between two strings can be done in a
privacy-preserving manner, as well, using privacy-preserving set
intersection protocols (e.g., 
see~\cite{ae-nepps-07,FNP04,vc-ssica-05,ss-ppsip-07,ss-ppsib-08})
or SMC methods for
dot products (e.g., see~\cite{dfknt-uscfm-06,gmw-hpamg-87,y-psc-82}).
In addition, the Fairplay system~\cite{bnp-fssmp-08} provides a general
compiler for building such computations.

Du and Atallah~\cite{da-psrda-01} study an SMC protocol
for querying a string $Q$ 
in a database of strings, ${\cal X}$, as in our framework,
where comparisons are based on approximate
matching (but not sequence-alignment).
Their SMC protocols for performing such queries provide
a best match, not a score for each string in the database.
Thus, their scheme would not be applicable in the attack framework we are
considering in this paper.
The SMC method of Jiang {\it et al.}~\cite{jmcs-sddli-08}, on the other hand,
is directly applicable.
It provides a vector of scores comparing a string (or vector) 
$Q$ to a sequence of
strings (or vectors), as we require in this paper.
Thus, our Mastermind methods can be viewed as an attack on repeated use of the
SMC protocol of Jiang {\it et al.}

Goodrich~\cite{g-magd-09} studies the 
problem of discovering a single DNA string from 
a series of genomic Mastermind queries.
All his methods are sequential and adaptive, however, so the only way they
could be applied to the data-cloning attack on an entire biological database is
if Bob were to focus on each string $X_i$ in ${\cal X}$ in turn. 
That is, he would have to gear his
queries to specifically discovering each $X_i$ 
in $n$ distinct ``rounds''
of computation, each of which uses a lot of string-comparison queries. 
Such an adaptation of Goodrich's Mastermind attacks to perform data cloning,
therefore, would be
prohibitively expensive for Bob. 
Our approach, instead, is based on performing a nonadaptive Mastermind
attack on the entire database at once.

We note that others have investigated de-anonymization techniques on both
social networks~\cite{backstrom2007} and Netflix data~\cite{narayanan2008}.
These works are complementary to our goal of cloning the databases themselves.

\section{Exploiting Sparsity in an Algorithmic Data-Cloning Attack}
\label{sec:sparsity}
In this section, we describe the details of our nonadapative Mastermind data-cloning attack.
It is often the case that all or a large fraction of the strings in a
real-world string database can be 
characterized in terms of a small number of differences with a public
reference string. In these situations, which are quite common, we can
apply a reduction to nonadaptive group testing, which results in an
an efficient Mastermind attack as we will see.

\subsection{Non-adaptive Combinatorial Group Testing}
\emph{Group testing} 
was introduced by Dorfman~\cite{d-ddmlp-43}, during World War II,
to test blood samples.
The problem he addressed was to design an efficient way to detect the
few thousand blood samples that were contaminated with syphilis out
of the millions that were collected.
His idea was to pool drops of blood from multiple samples and test
each pool for the syphilis antigen. By carefully arranging the
group tests and then discovering which groups
tested positive and which ones tested negative 
he could then identify the contaminated
samples using a small number of group tests (much
smaller than the number needed to explicitly test each individual blood sample),
thereby sparing thousands of G.I.'s from needless disease exposure.
In this paper, we show that Dorfman's humanitarian discovery has an
unfortunate dark side when it comes to privacy protection, 
for it enables privacy leaks to be amplified in a data-cloning attack.

In the \emph{combinatorial group testing} problem
(e.g., Du and Hwang~\cite{dh-cgtia-00}), one is given a set $S$ of
$n$ items, at most $d$ of which are ``defective,''
for some parameter $d\le n$, and one is interested in
exactly determining which of the items in $S$ are defective.
One can form a test from any subset $T$ of $S$ and in a single step
determine if $T$ contains any defective items or not.
If one can use information from the result of a test in formulating
the tests to make in the future, then the method is said to be
\emph{adaptive}. If, on the other hand, one cannot use the results
from one test to determine the makeup of any future test, then the
method is said to be \emph{nonadaptive}.
For the application to the Mastermind attack, we are interested in
nonadaptive methods.

There are several existing nonadaptive group testing methods~\cite{dh-cgtia-00}, 
but these approaches are meant for a more general context than 
in our database cloning attack.
In particular, these methods are designed to work for \emph{any} set
of items having $d$ defective members.
In our case, we are instead interested in specific sets of items that
are derived from the database we are interested in cloning.
Because of this, we can, in fact, derive improved bounds than would
be implied by existing combinatorial group-testing methods.

Suppose we are given a collection, $\cal C$,
of sets, 
$
{\cal C} = \{S_1,S_2,\ldots, S_g\}$,
which are not necessarily distinct, such that each set $S_i$
contains $n$ items, at most $d$ of which are ``defective.''
We want to design a nonadaptive group testing scheme that can
exactly identify the subset, $D_i$, of
at most $d$ defective items in each set $S_i$ in $\cal C$.
Our approach to solving this problem
is based on a 
randomized approach used by Eppstein {\it et al.}~\cite{egh-icgtr-05}.

A nonadaptive group testing algorithm
can actually be viewed as a $t \times n$ 0-1 matrix, $M$.
Each of the $n$ columns of $M$ corresponds to one of the $n$ items
and each of the $t$ rows of $M$ represents a test. 
If $M[i,j]=1$, then item $j$ is included in test $i$, and
if $M[i,j]=0$, then item $j$ is not included in test $i$.
Since this is a nonadaptive testing scheme, 
we assume that
no test depends on the results of any other.
That is, every row of the matrix $M$ is defined in advance of any
test outcomes.
The analysis question, then, is to determine how large $t$ must be 
for the results of these tests to provide useful results.

Let $C$ denote the set of columns of $M$.
Given a subset $D$ of $d$ columns in $M$, and a specific column $j$
in $C$ but not in $D$,
we say that $j$ is \emph{distinguishable} from $D$ if there is a row
$i$ of $M$ such that $M[i,j]=1$ but $i$ contains a $0$ in each of the
columns in $D$.
If each column of $M$ that is in $C$ and not in $D$ is 
distinguishable from $D$, then we say that $M$ is 
\emph{$D$-distinguishing}.
Furthermore, we generalize this definition, so that
if $M$ is $D_i$-distinguishing for each subset, $D_i$, in a collection,
$
{\cal D}=\{D_1,D_2,\ldots,D_g\}$, 
of columns in $C$,
then we say that $M$ is $\cal D$-distinguished.
Finally, we say that the matrix $M$ is $d$-{\it disjunct} 
(e.g., see Du and Hwang~\cite{dh-cgtia-00}, p.~165)
if it is $\cal D$-distinguished for the collection, $\cal D$,
of all of the 
$
{n \choose d}
$
subsets of size
$d$ of $C$.

Note that if $M$ is $D$-distinguishing, 
then it leads to a simple testing
algorithm with respect to $D$.
In particular, suppose $D$ is the set of defective items
and we perform all the tests in $M$.
Note that, since $M$ is $D$-distinguishing,
if an item $j$ is not in $D$, then there is a test
in $M$ that will determine the item 
$j$ is not defective, for $j$ would belong to a test that must
necessarily have no defective items.
So we can identify $D$ in this case---the set $D$ consists of
all items that have no test determining them to be nondefective.

Of course, if $M$ is $d$-disjunct, then this simple detection algorithm works 
for any set $D$ of up to $d$ defective items in $C$.
Unfortunately, building such a matrix $M$ that is $d$-disjunct
requires $M$ to have $\Omega(d^2\log n/\log d)$ 
rows~\cite{dh-cgtia-00,R94}.
So we will instead build a matrix that is $\cal D$-distinguished for
the collection, $\cal D$, of defective subsets determined by the sets
of items in $\cal C$, with high probability.

Given a parameter $t$, which is a multiple of $d$,
we construct a $2t\times n$ matrix $M$ as follows.
For each column $j$ of $M$, we choose $t/d$
rows uniformly at random and we set the values of these entries to $1$, with
the other entries in column $j$ being set to $0$.
Note, then, that for any set $D$ of up to $d$ defective items,
there are at most $t$ tests that will have positive outcomes
(detecting defectives) and,
therefore, at least $t$ tests that will have negative outcomes.
Our desire, of course, is for 
columns that correspond to samples that are distinguishable from
the defectives ones should belong to at least one negative-outcome test.
So, let us focus on bounds for $t$ that allow for such a matrix $M$
to be chosen with high probability.

Let $C$ be a set of (column) items having a fixed 
subset $D$ of $d$ defective items.
For each (column) item $j$ in $C$ but not in $D$, let $Y_j$ denote the
0-1 random variable that is $1$ if $j$ is falsely identified as
a defective item by $M$ (that is, $j$ is not included in a test of
items distinguished from those in $D$).
Let $Y_j$ be $0$ otherwise.
Observe that the $Y_j$'s are independent, since $Y_j$ depends only on
whether the choice of rows we picked for column $j$ collide with the
at most $t$ rows of $M$ picked for the columns corresponding
to items in $D$.
There are a total of $2t$ rows, at most $t$ of which contain a test
with a defective item.
Thus, the probability of any non-defective item joining any
particular test having a defective item in it is at most $1/2$;
hence, any $Y_j$ is $1$ (a false positive)
with probability at most $2^{-t/d}$, since each item is included in 
$t/d$ tests at random.

Let $Y=\sum_{j=1}^n Y_j$, and note that
the expected value of $Y$, $E(Y)$, is at most ${\hat\mu}=n/2^{t/d}$.
Thus, if ${\hat\mu}\le 1$,
we can use Markov's inequality to bound
the probability of the (bad) case when $Y$ is non-zero as follows:
\[
\Pr(Y\ge 1) \le E(Y) \le {{\hat\mu}} = \frac{{n}}{2^{t/d}} .
\]
Thus, if we set
\[
t\ge 2d\log n,
\]
then $M$ will be $D$-distinguishing with probability
at least $1-1/n$, for any particular subset of defective items, $D$,
from a set $C$ of $n$ items.
Likewise, if we set 
\[
t\ge 2d\log n+d\log g,
\]
then $M$ will be $\cal D$-distinguished, with probability at least $1-1/n$,
for the collection of $g$ subsets of defective items determined by
the sets in $\cal C$.
Finally, we can use the fact 
(e.g., see Knuth~\cite{k-acp-73})
that 
\[
{n\choose d}<(en/d)^d,
\]
so that
if we set 
\[
t\ge 2d\log n + d^2\log (en/d),
\]
then $M$ will be
$d$-disjunct with probability at least $1-1/n$, which implies $M$ will work
for any subset of at most $d$ defective items.
Therefore, we have the following.

\begin{theorem}
\label{thm:math}
If 
\[
t \ge 2d\log n + \min\{d\log g, d^2\log (en/d)\},
\]
then a $2t\times n$ random matrix $M$,
constructed as described above,
is $\cal D$-distinguished, with 
probability at least $1-1/n$, for any given collection, 
${\cal D}=\{D_1,D_2,\ldots,D_g\}$, of $g$
subsets of size $d$ of the $n$ columns in $M$.
\end{theorem}
\begin{proof}
Let $\cal D$ be a given collection of $g$ (not necessarily distinct)
subsets of size $d$ of the $n$ columns in $M$.
If 
\[
d^2\log (en/d) > d\log g, 
\]
then $M$ is $\cal D$-distinguished by construction, with probability
at least $1-1/n$.
If, on the other hand,
\[
d^2\log (en/d)\le d\log g,
\]
then $M$ constructed as above
is $d$-disjunct, with probability at least $1-1/n$,
which implies it is $\cal D$-distinguished w.h.p.~for any 
collection $\cal D$ of subsets of size $d$ 
of the $n$ columns of $M$.
\end{proof}

As mentioned above,
this is a way of constructing a simple nonadaptive group testing method for
identifying the defective items in the collection, $\cal D$, of
subsets of up to $d$ defective items determined by the sets in $\cal C$.

\subsection{Reducing Mastermind to Group Testing}
In this section,
we describe how to use nonadaptive group testing to
construct an efficient Mastermind cloning attack.
\label{sec:arb}
Consider the case when ${\cal X}$ is a database of $g$ strings of length
$n$ each, with each of them having
at most $d\le n$ differences with a reference string $R$.
We assume that each string in ${\cal X}$ is drawn from an alphabet of $c$
characters, which, without loss of generality, we assume are integers in the
range $[0,c-1]$.

Suppose, like before, that we have a $2t\times n$ 
nonadaptive group testing matrix, $M$, for a 
set of size $n$ having at most $d$ defectives,
where 
\[
t\ge 2d\log n + \min\{\log g, d^2\log (en/d)\}.
\]
As before,
we begin our general Mastermind cloning attack by making a query 
for the reference string, $R$.
Let $r$ be the response score for the query for $R$.
Next, we create $c-1$ different string queries, $Q_{k,l}$ 
for each of the $2t$ tests in $M$, defined, for $l=1,2,\ldots,c-1$, as
follows:
\[
Q_{k,l}[j] = \left\{ \begin{array}{ll}
		R[j] & \mbox{if $M[k,j]=0$} \\
		(R[j] + l) \bmod c & \mbox{else.}
		\end{array}
		\right.
\]
Each such query against a string $X_i$ will have some response, $r_{k,l,i}$.
We interpret test $(k,l,i)$ as having a ``positive'' response, that is, it
does not detect a defective, if, in making the comparison of $Q_{k,l}$ with
the string $X_i$, the response
\[
r_{k,l,i}=r-b_{k,0,i},
\]
where $b_{k,0,i}$ is the number of
characters in $X_i$ matching their associated (color-$0$) location in $R$
at places where there are $1$'s in row $k$ of $M$.
Intuitively, each $1$ in row $k$ of $M$ indicates a place where we test
a deviation from the reference value in $R$ at that location to 
the color $l$ away. If none of these
locations is a match with the current $X_i$ string, then none of these 
locations take a color that is $l$ additive colors from their reference value. 
In other words, defective ``items'' in the
associated group testing method 
correspond to locations where $X_i$ differs from the reference string
with characters that are exactly $l$ away from their reference values. 

Of course, being able to determine if such a test for $Q_{k,l}$ against
string $X_i$ is ``positive'' or
``negative'' requires that we know the value $b_{k,0,i}$, which we don't 
immediately know.
We do immediately know the number, $b_k$, 
of $1$'s in row $k$ of $M$, however. And, after
we perform the queries for each $Q_{k,l}$ against a string $X_i$, 
we learn each response $r_{k,l,i}$.
That is, we have $c$ linear equations in $c$ unknowns from these
queries and their responses.  
Specifically, we have the equation,
$
b_k = b_{k,0,i} + b_{k,1,i} + \cdots + b_{k,c-1,i}$,
where $b_{k,l,i}$ denotes the number of places $j$ where there is a 
$1$ in row $k$ of $M$ and the character in position $j$ of $X_i$ is 
$l$ away from the reference, that is, places
where $X_[j] = (R[j] + l)\bmod c$ and $M[k,j]=1$.
We also have $c-1$ equations,
\[
r_{k,l,i} = r - b_{k,0,i} + b_{k,l,i},
\]
for $l=1,2,\ldots,c-1$,
which can each be rewritten as
$
b_{k,l,i} = r_{k,l,i} - r + b_{k,0,i} $.
This allows us to rewrite
\[
b_k = c\,b_{k,0,i} - (c-1)r + \sum_{l=1}^{c-1} r_{k,l,i}.
\]
Thus, we can determine the value of $b_{k,0,i}$ as
\[
b_{k,0,i} = \frac{b_k + (c-1)r - \sum_{l=1}^{c-1} r_{k,l,i}}{c} ,
\]
which in turn tells us which of the $Q_{k,l}$ tests are ``positive'' and
which ones are ``negative.''
Essentially, we are performing a combinatorial group test for each possible
shift we can make from a color in reference $R$.

Thus, if there are at most $d$ locations where $X_i$ differs from the
reference string and $M$ is $\cal D$-distinguished for the set of at most $d$
locations of difference for each string in ${\cal X}$, then this scheme will learn
the complete identity of each string in ${\cal X}$.
That is, this method will clone ${\cal X}$, with high probability.
Therefore, by Theorem~\ref{thm:math}, we have the following:

\begin{theorem}
Given a database
$ {\cal X} = (X_1,X_2,\ldots,X_g)$,
of strings of length $n$ defined over an alphabet of size $c$,
there is a nonadaptive Mastermind cloning method that can discover each
string in $\cal X$,
using $2t(c-1)$ tests, with probability at least $1-(c-1)/n$, where
$ t $ is the smallest multiple of $d$ such that 
\[
t\ge 2d\log n + \min\{d\log g, d^2\log (en/d)\},
\]
and $d\le n$ is
the maximum number of differences any string in ${\cal X}$ has with $R$.
\label{thm:sparse}
\end{theorem}

\section{Case Studies}
\label{sec:experiments}

To test the real-world risks of the nonadaptive Mastermind cloning attack, we applied 
our methods to case studies involving random samples from
a number of real-world string and vector databases,
including genomic and social network data.  
We briefly describe the data sets 
used and then discuss experimental results which reveal that relatively few 
tests are needed to recover large proportions of each database. 

\begin{table}[hbt!]
\centering
\begin{tabular}{|l|c|c|c|c|c|c|c|}
\multicolumn{1}{l}{ Name}  
&\multicolumn{1}{c}{Strings} 
&\multicolumn{1}{c}{Length}
&\multicolumn{1}{c}{Max Diff}
&\multicolumn{1}{c}{Colors}  \\
\hline
Genomic & 457 & 16,568 & 492 & 4 \\
Netflix & 1,000 & 17,700 & 1988 & 6 \\
Epinions & 2,000 & 131,827 & 517& 3\\
Slashdot & 2,000 & 82,144 & 378 & 3\\
Slashdot (All) & 82,144 & 82,144 & 428 & 3\\
Facebook-UNC & 18,163 & 18,163 & 3,795 & 2\\
Facebook-Unif & 1,000 & 72,261,577 & 2,164 & 2\\
\hline
\end{tabular}
\end{table}

\begin{figure}[hbt!]
\centering
\quad\includegraphics[scale=0.3]{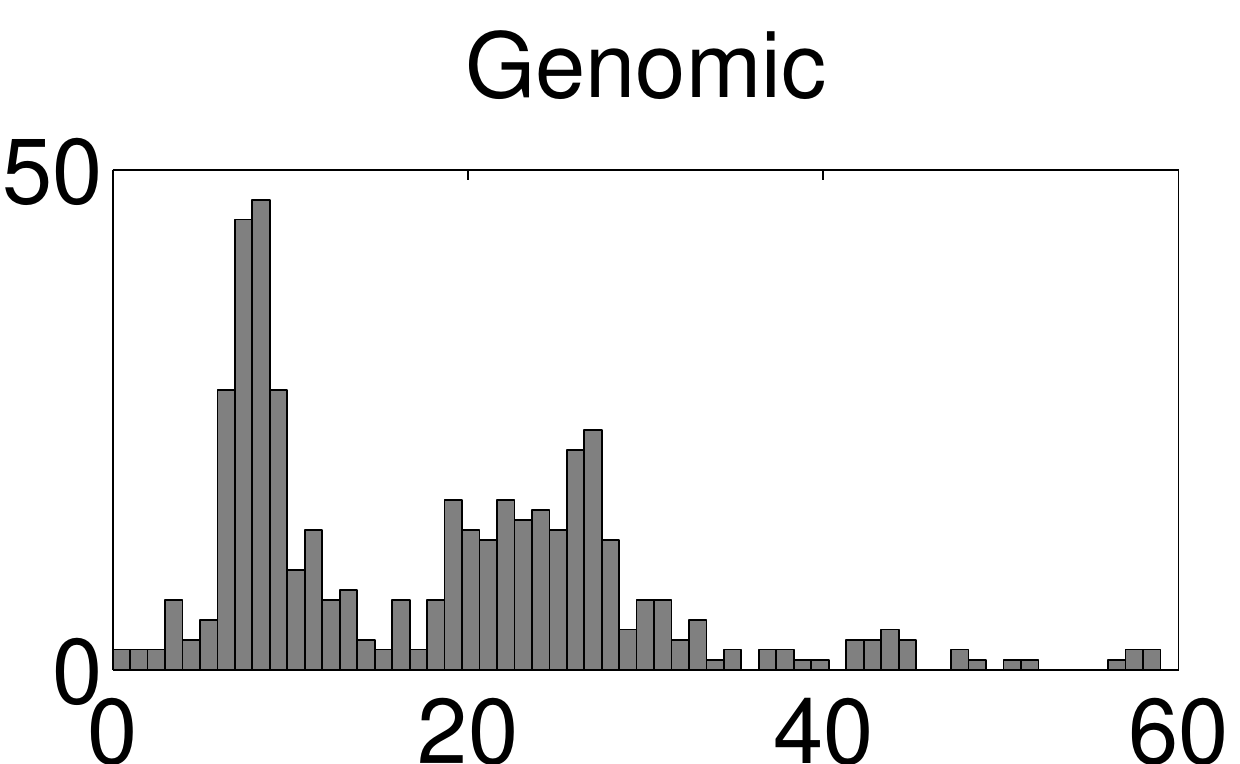} \quad
\includegraphics[scale=0.3]{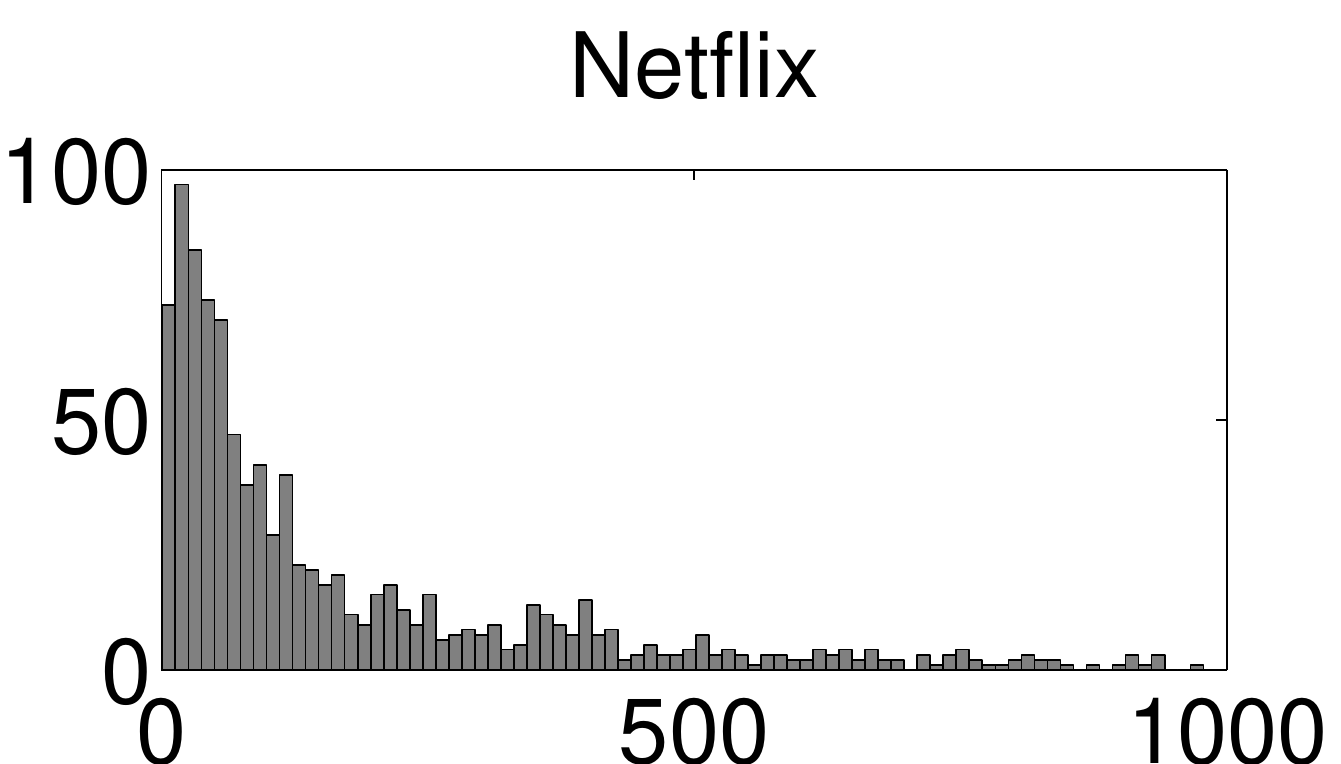} \\ \quad \\
\includegraphics[scale=0.3]{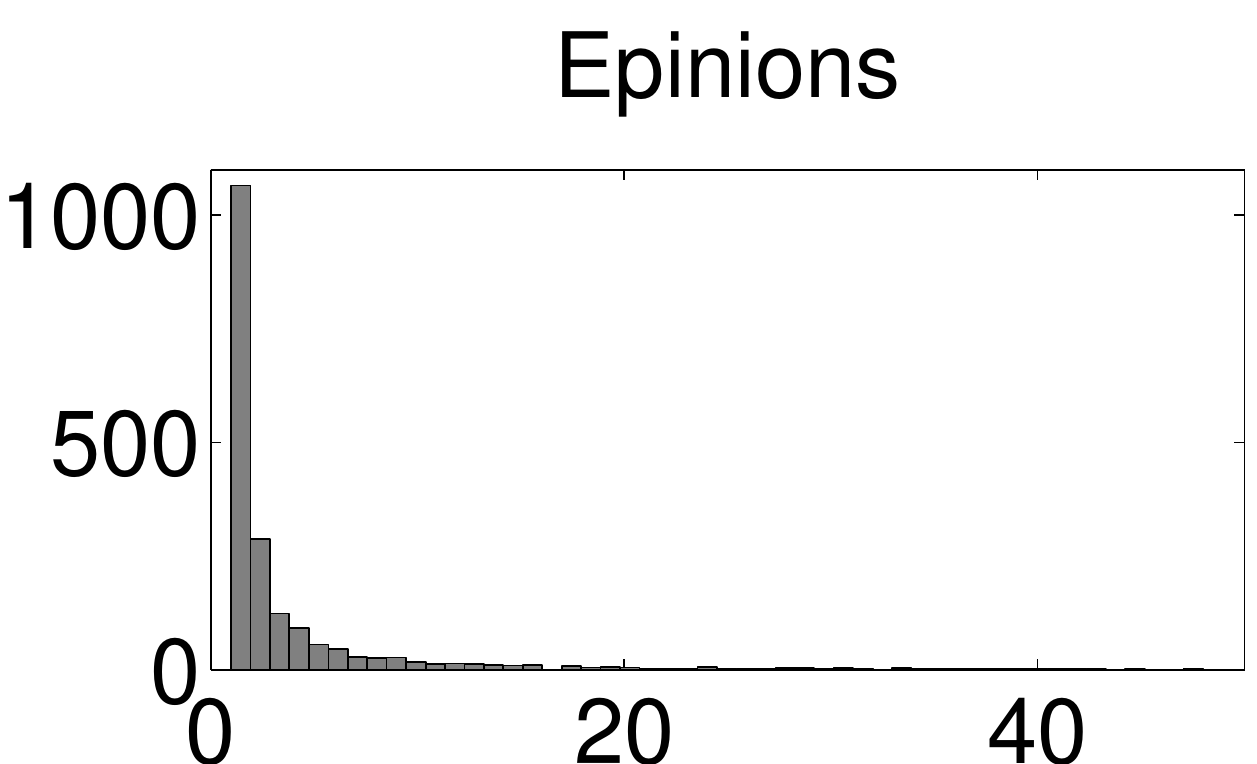} \qquad
\includegraphics[scale=0.3]{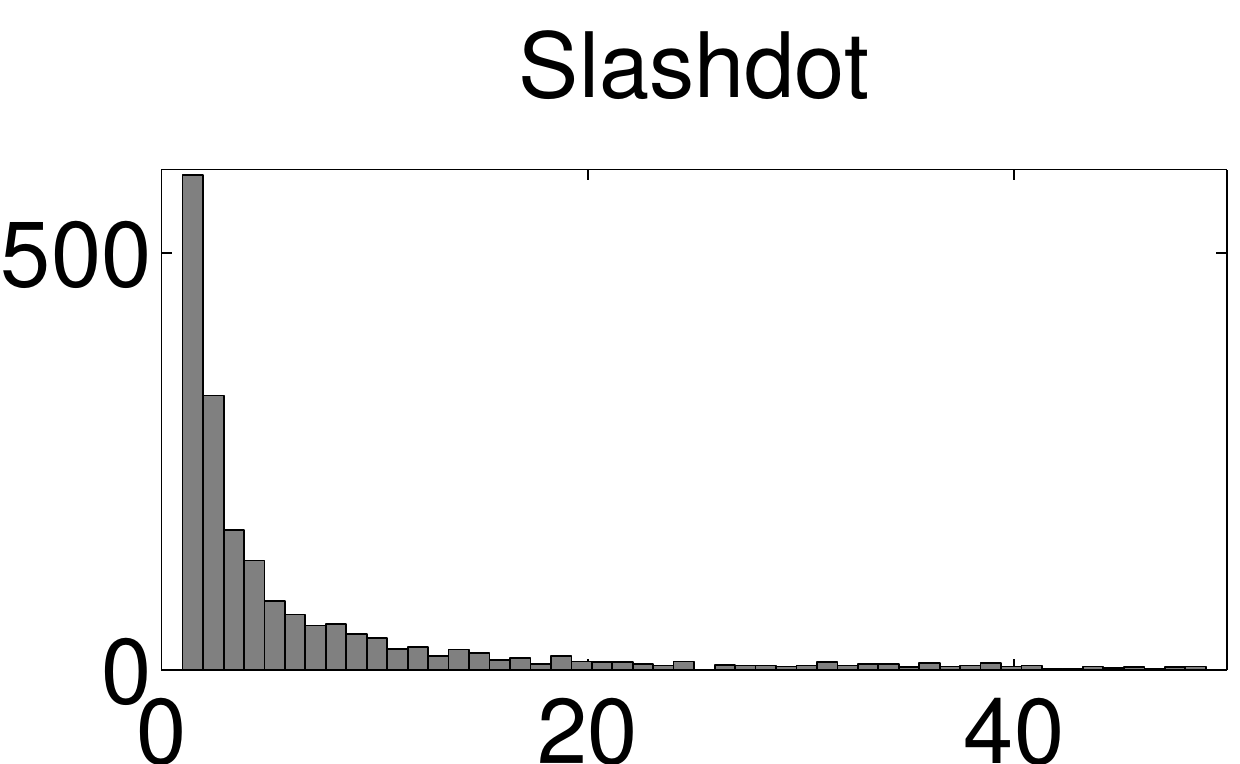} \\ \quad \\
\includegraphics[scale=0.3]{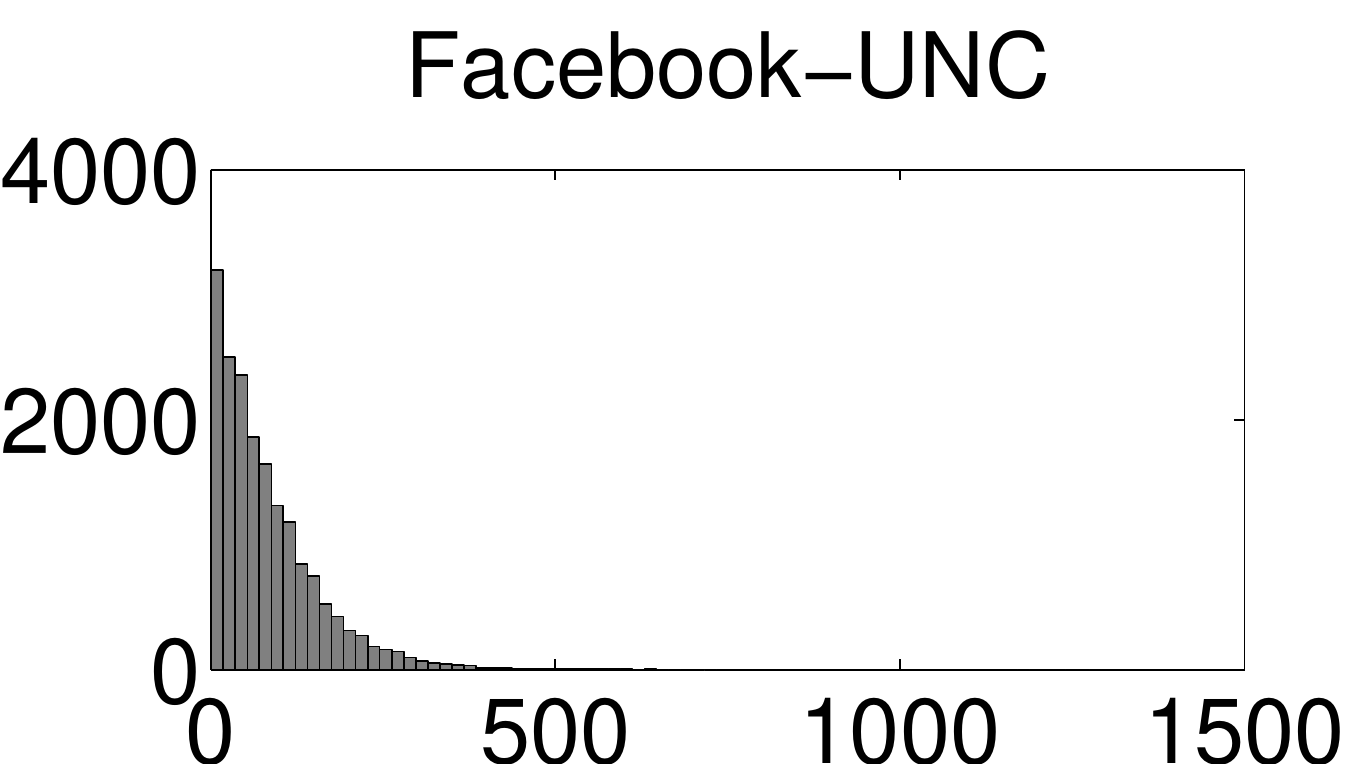} \quad
\includegraphics[scale=0.3]{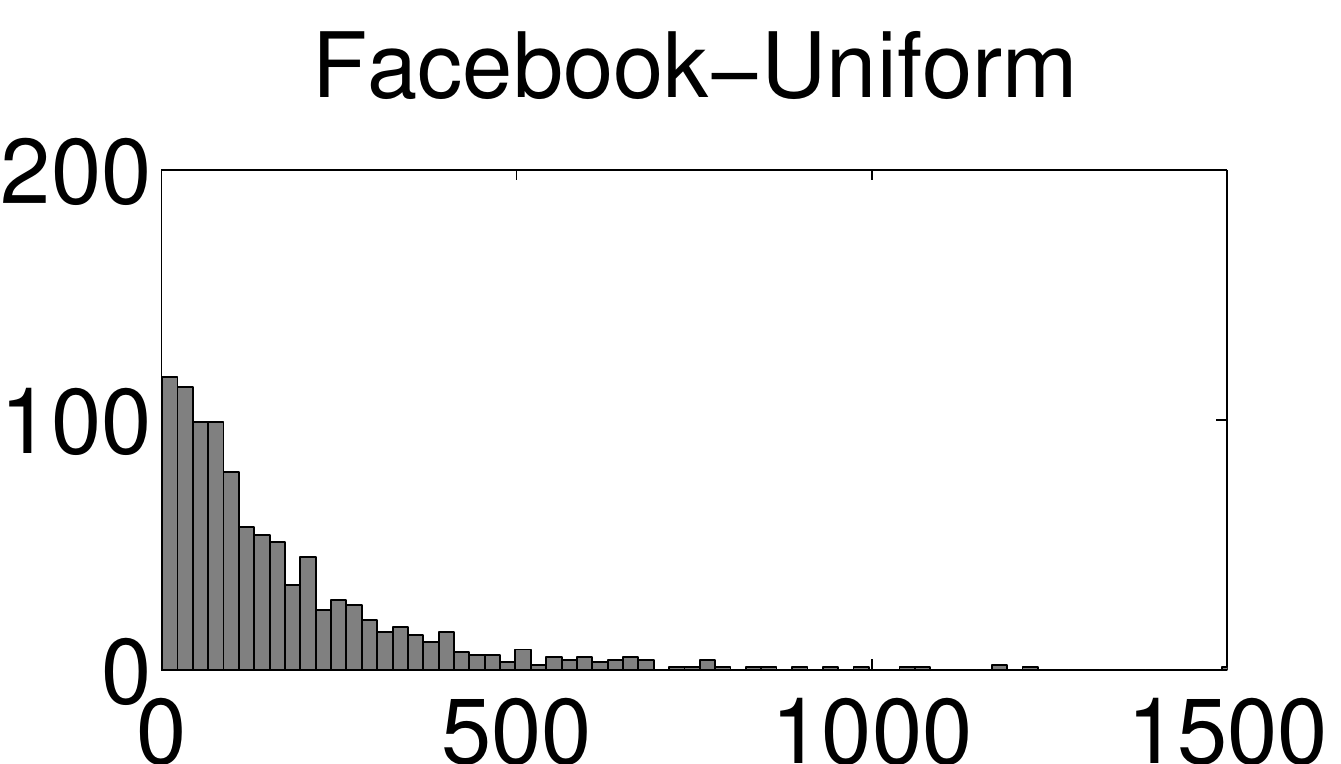}
\caption{Data sets used in experiments, along with histograms of differences from reference $R$.
These data sets have varying characteristics.}
\label{tbl:data}
\end{figure}

\subsection{Data Sets}

We analyze several different data sets with varying characteristics to test our approach.
For each data set, Figure \ref{tbl:data} shows the number of strings $g$, string length $n$, 
maximum difference $d$ from the reference $R$ across strings (where ``difference'' is 
defined as the number of entries in which the string differs from $R$), 
and the number of unique colors $c$ present in the database.

The Genomic database consists of 457 human mitochondrial sequences 
downloadable from GenBank\footnote{\url{http://www.ncbi.nlm.nih.gov/Genbank/index.html}}.  
We use the Revised Cambridge Reference Sequence (rCRS), of length 16,586 bp 
as the reference string $R$.  Figure \ref{tbl:data} shows the distribution of 
sequence differences from $R$, which reveals that the differences from $R$
are relatively few and are concentrated at several different modes.  In this data,
there are four colors, namely the nucleotides A, C, T, and G.

Our movie-rating database is taken from the Netflix Prize
data\footnote{\url{http://www.netflixprize.com}}, 
which consists of 100 million movie ratings and 480,189 different Netflix 
users.  In our experiments, we use a representative subset of 1,000 randomly selected 
users.  Each user has an associated string over 17,770 movies, where each position 
$i$ stores the rating (from 1 to 5) given by the user for movie $i$.  An entry of 0
signifies that the user has not rated that movie.  Thus, there are six different unique
colors in this database (0-5).  Our reference string $R$ consists of all zeros, 
representing the case where no movies are rated.  According to 
Figure \ref{tbl:data}, the majority of users rate less than 300 movies.  This 
sparsity allows the group testing attacks to be very efficient, as we will see in 
the experiments.
 
We also analyze online social networks such as Epinions, Slashdot, and Facebook. 
Available from the SNAP Library\footnote{\url{http://snap.stanford.edu/data/index.html#signnets}},
Epinions and Slashdot are ``signed'' networks, where positive and 
negative links appear in the network's adjacency matrix \cite{Leskovec2010}.  The 
Epinions network is the site's ``Web of Trust'' where users specify the other users that they 
trust or distrust.  Similarly, in the Slashdot network, users can specify both ``friends'' and ``foes''.  
Hence, in both these databases, there are three unique colors: 0 (no
link), 1 (good link), and -1 (bad link).  In our experiments for both Epinions and Slashdot, 
we select a random subset of 2,000 users and utilize the corresponding rows in the adjacency matrix 
as our database.  
We also simulate a single large-scale group testing attack on the 
entire Slashdot-All adjaency matrix with 82,144 users. 
 
The two Facebook data sets that we analyze are Facebook-Uniform and Facebook-UNC.
Facebook-Uniform, provided by the authors of \cite{facebook}, is an unbiased sample
of 957K unique users obtained by performing Metropolis-Hastings random walks over the Facebook network. 
Each user is associated with a (sparse) binary vector of size 72 million which denotes adjacencies.  
We restrict ourselves to a random subset of 1,000 users in Facebook-Uniform.  Meanwhile, Facebook-UNC 
is a self-contained Facebook network of approximately 18,000 students at the 
University of North Carolina at Chapel Hill \cite{facebook2}.  

For all the social network data sets, we use a reference
string $R$ of all zeros.
Figure \ref{tbl:data} shows that these networks are also sparse, which is often the case in many real-world settings. 

\subsection{Experiments}

Our experimental approach is based on the analysis in Section \ref{sec:sparsity}.  Similar to 
randomly selecting $\frac{t}{d}$ rows from $2t$ rows (for each column in the nonadaptive 
group matrix $M$), we stochastically set each entry in $M$ to 1 with 
probability $p = \frac{1}{2d}$.  This procedure enables us to add additional tests 
to $M$ until the string is cloned or until a cutoff of $100,000 * c$ tests is reached, where
$c$ is the number of unique colors in the database.  We initialize with the same 
random seed for each string, ensuring that the same exact tests are performed on 
each string.  This scheme allows us to determine the actual number
of tests needed to clone the strings.  

\begin{table}
\caption{Theoretical number of tests needed to clone entire database (a) by baseline method 
(b) by nonadaptive Mastermind attack.}
\begin{center}
\begin{tabular}{l|c|c}
 & Baseline & Mastermind \\
\hline
Genomic & 49,704 & 76,752 \\
Netflix & 88,500 & 536,760\\
Epinions & 263,654 & 66,176 \\
Slashdot & 164,288 & 46,872  \\
Slashdot (All) & 164,288 & 58,208 \\
Facebook-UNC & 18,163 & 227,700 \\
Facebook-Uniform & 72,261,577 & 190,432 \\
\end{tabular}
\end{center}
\label{tbl:theoretical}
\end{table}

Before delving into the experimental results, we report in Table \ref{tbl:theoretical} the 
theoretical number of tests needed to clone the entire database with high probability, 
using the nonadaptive Mastermind technique.  These numbers are based on $n$, $g$, $d$, $c$, and the 
bound in Theorem \ref{thm:sparse}.  Table \ref{tbl:theoretical} also shows the number of tests needed 
by a baseline technique to exactly clone the entire database.  This baseline technique
generates tests based on the reference $R$.  For each entry $j$ within $R$, 
and for each color offset $l$, a test is created where the entry $j$ in $R$ is replaced 
with its color offset $l$, namely $(R[j] + l)\mod c$.  Thus, the baseline method needs 
$(c-1) * n$ tests to recover the database.  Interestingly,
the baseline technique can beat the theoretical bound (with $d$) when $n$ is small, 
as is the case for the Genomic, Netflix, and Facebook-UNC data.

Fortunately, our Mastermind attack can take advantage of the sparsity in the 
data to improve its efficiency.  Since each string's distance from $R$ is usually much smaller 
than $d$, it is empirically advantageous to use a target $\hat{d}$ that is much smaller 
than $d$.  For instance, the Netflix data has a maximum difference $d=1988$, but the mean 
difference from $R$ is $d_{\text{mean}}=202$ and the median is 
$d_{\text{median}}=92$.  Thus, there are different possible choices for $\hat{d}$.

\begin{figure*}
\centering
\includegraphics[scale=0.32]{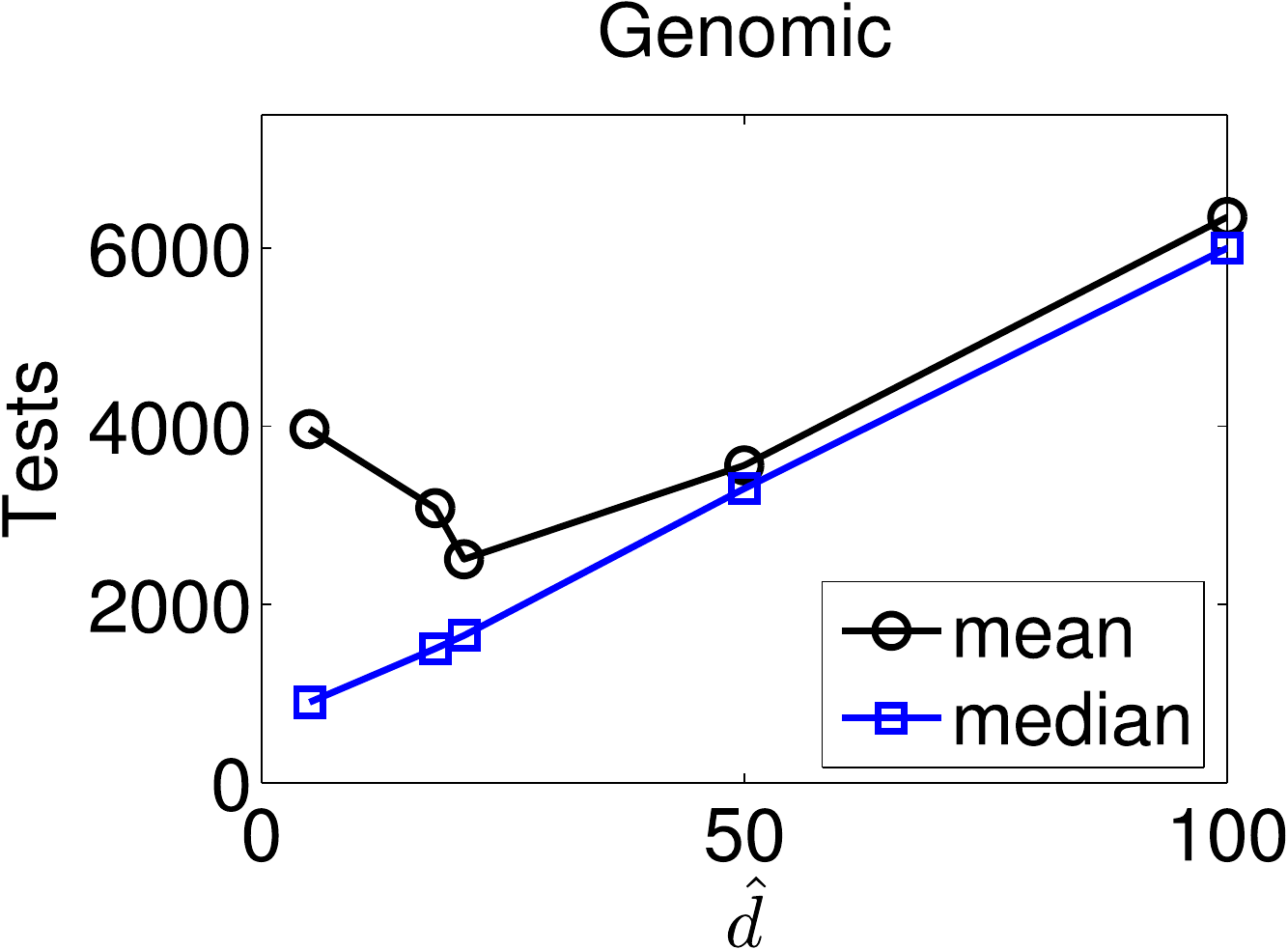} \qquad
\includegraphics[scale=0.32]{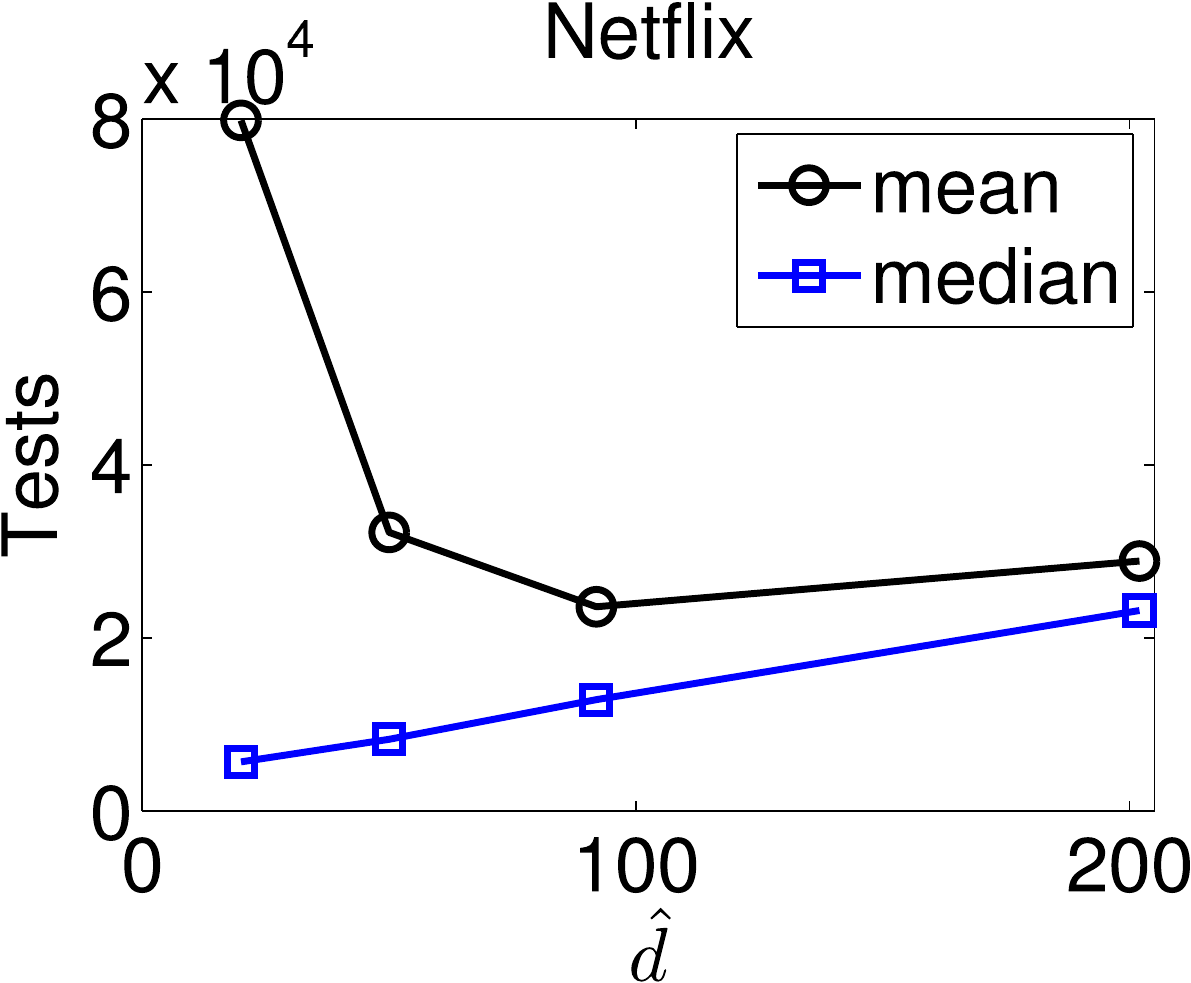} \\ \quad \\
\includegraphics[scale=0.32]{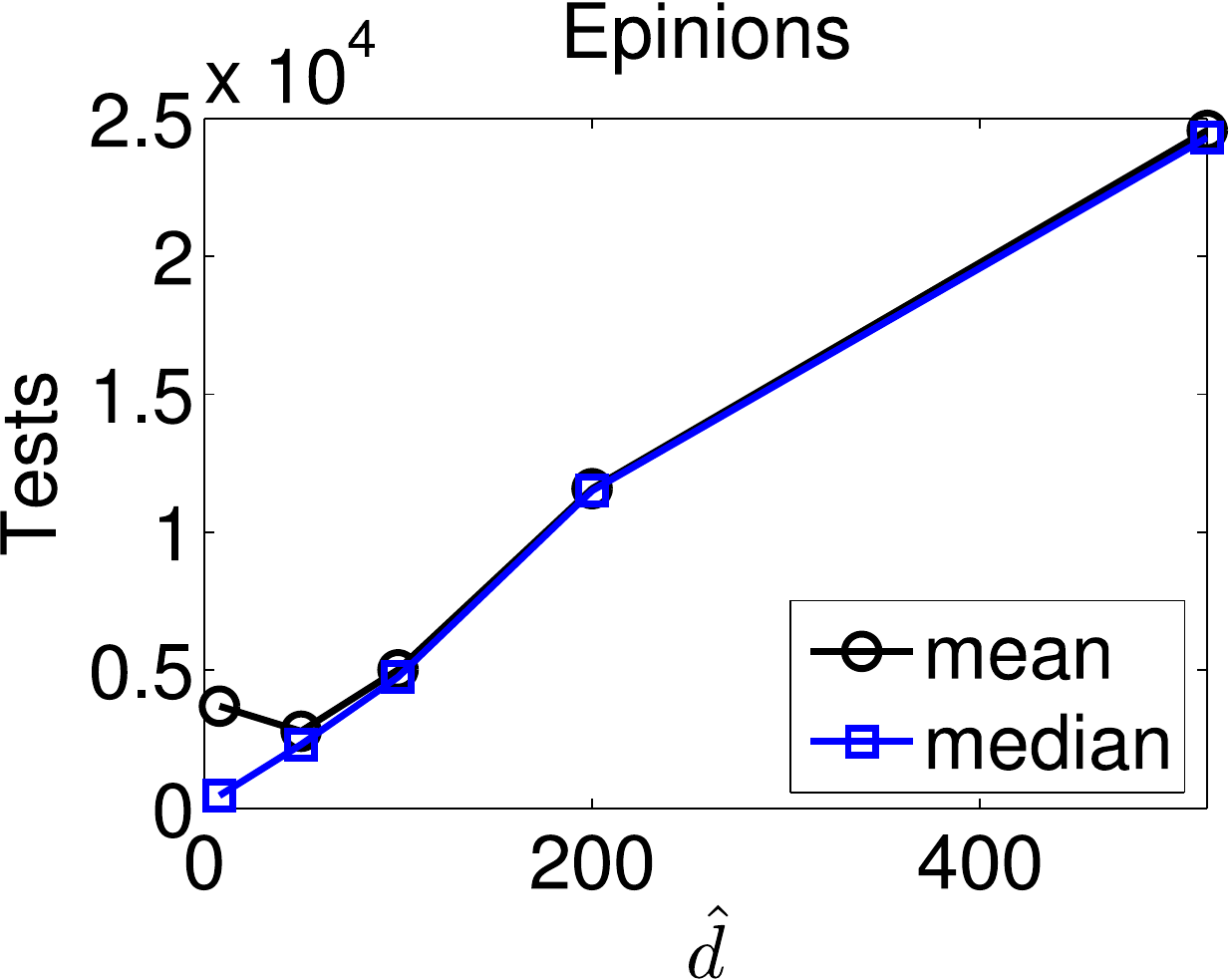} \quad
\includegraphics[scale=0.32]{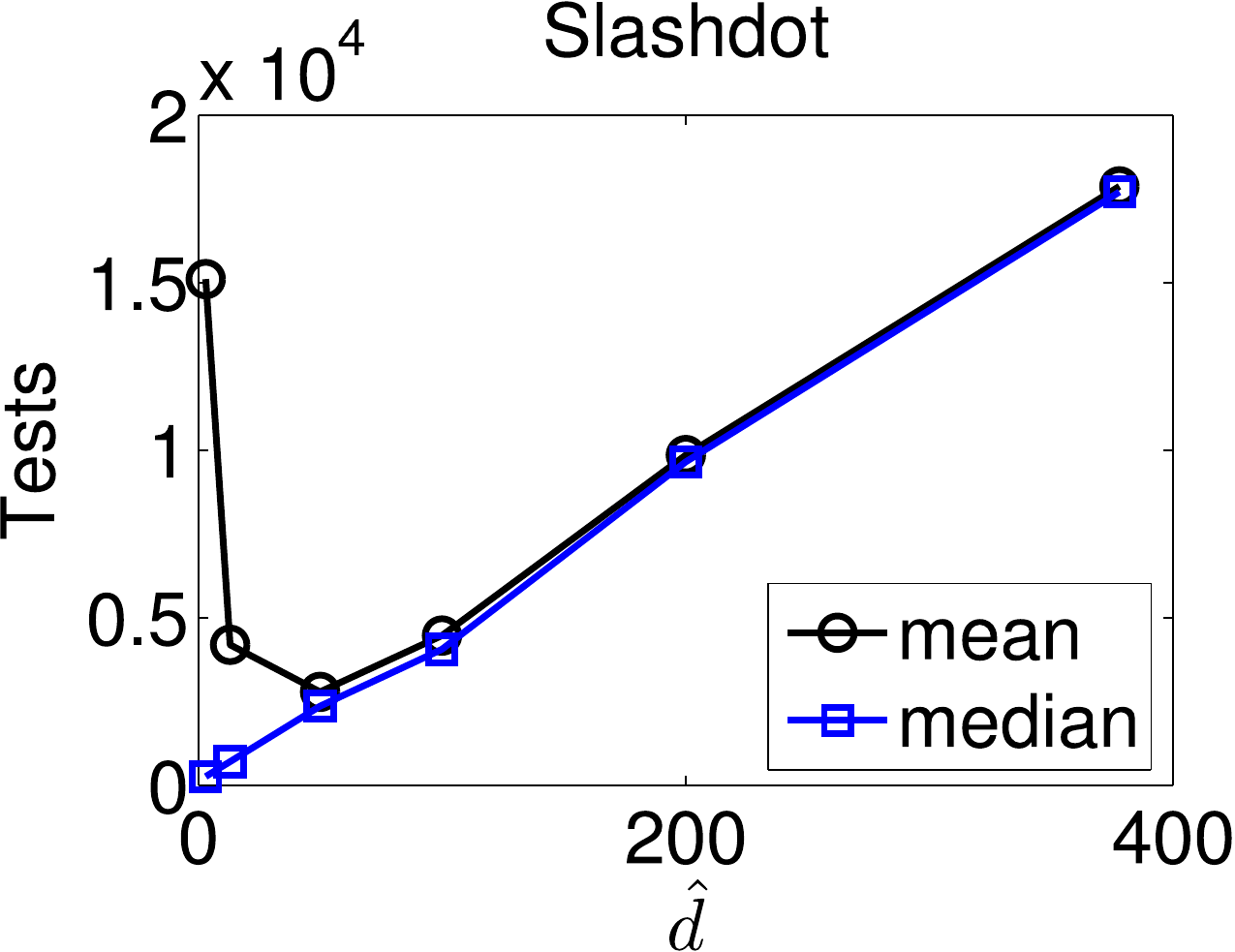} \quad
\includegraphics[scale=0.32]{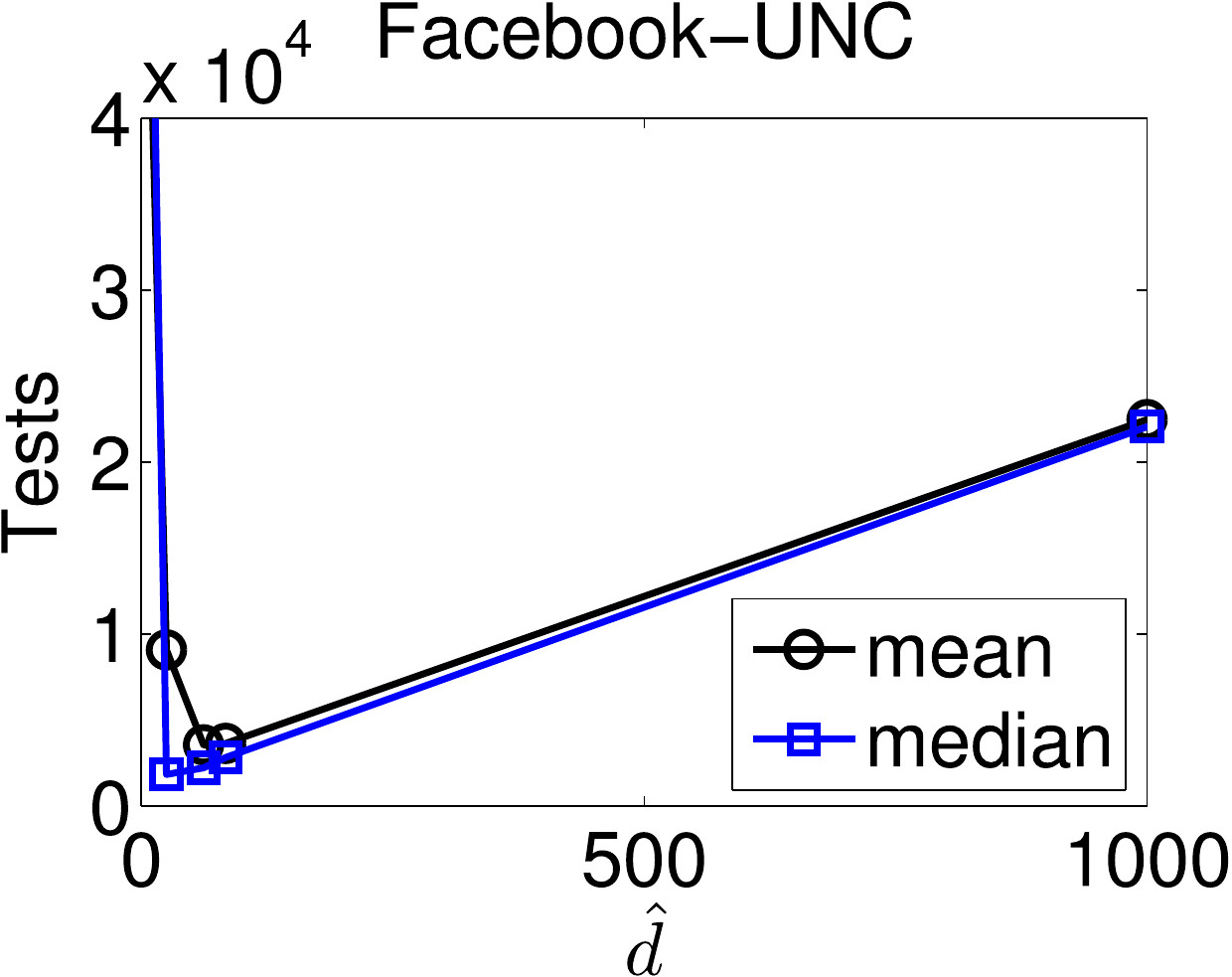} 
\caption{Mean and median number of tests required until string is cloned (averaged across all strings in database), for various settings of target distance $\hat{d}$. Typically, it is advantageous to 
set $\hat{d}$ to be much less than $d$, since most of the vectors are sparse and are close to the reference R.}
\label{fig:target}
\end{figure*}

\begin{figure*}
\centering
\includegraphics[scale=0.32]{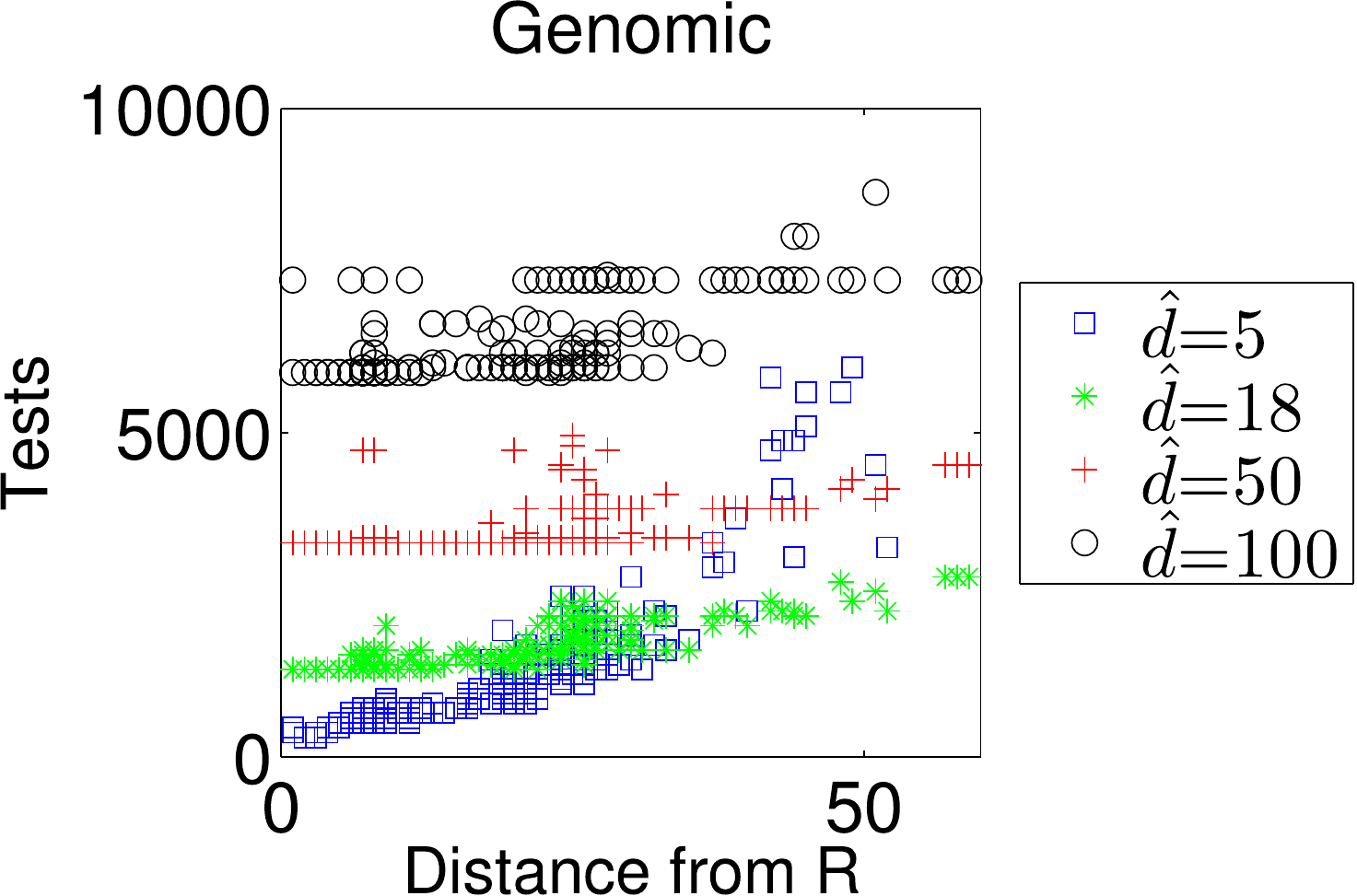} \qquad
\includegraphics[scale=0.32]{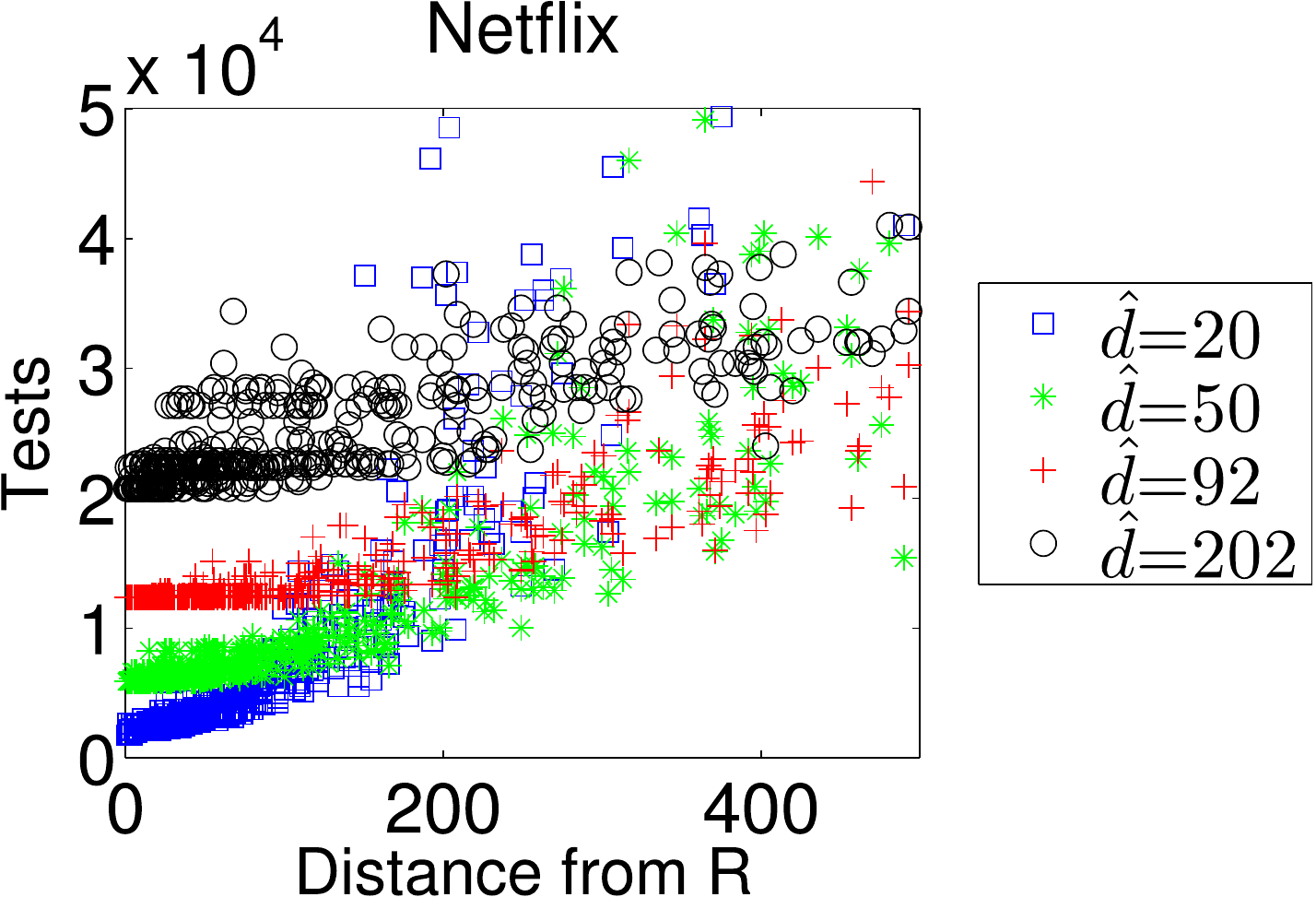} \\ \quad \\
\includegraphics[scale=0.32]{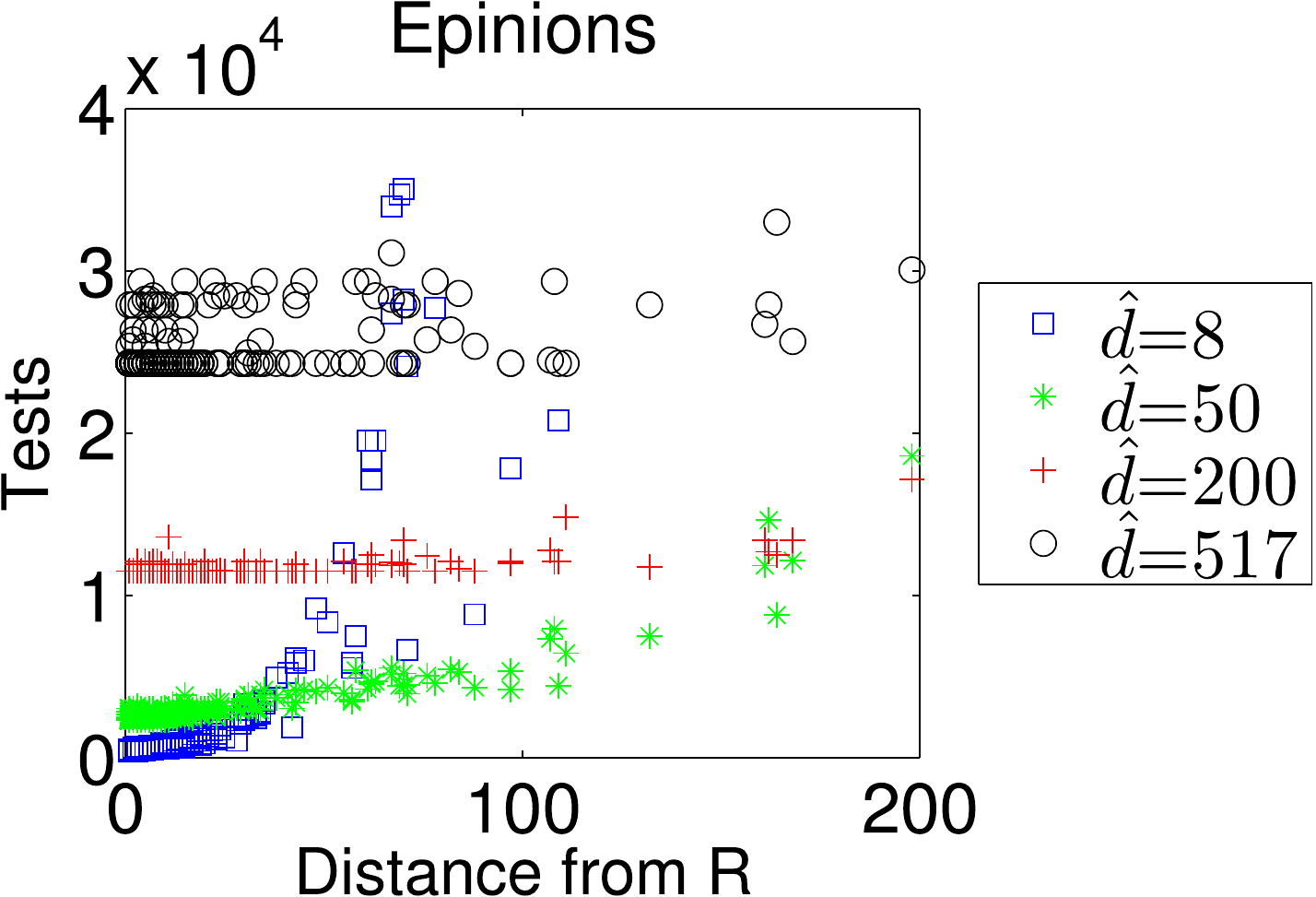} \quad
\includegraphics[scale=0.32]{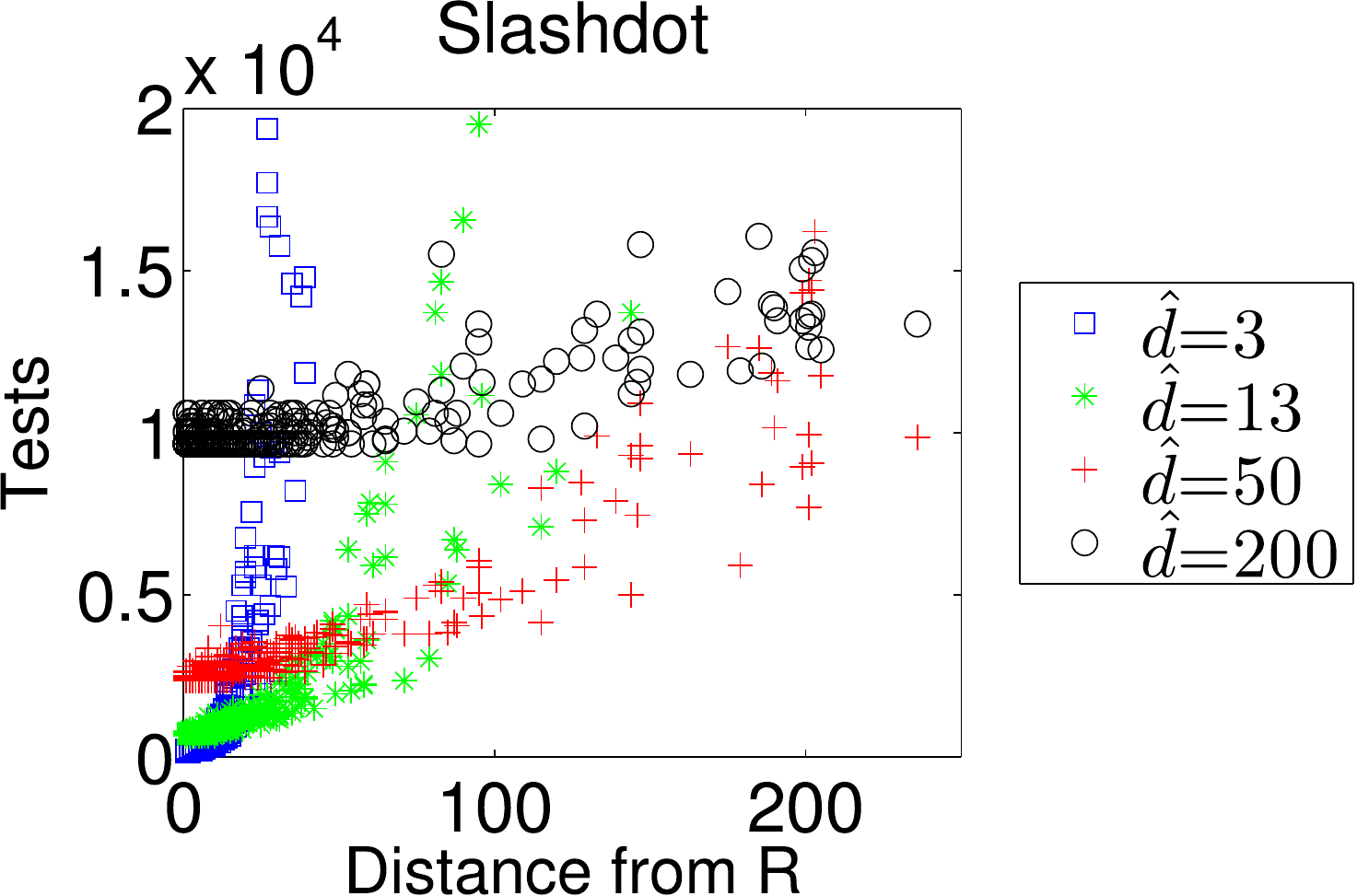} \quad
\includegraphics[scale=0.32]{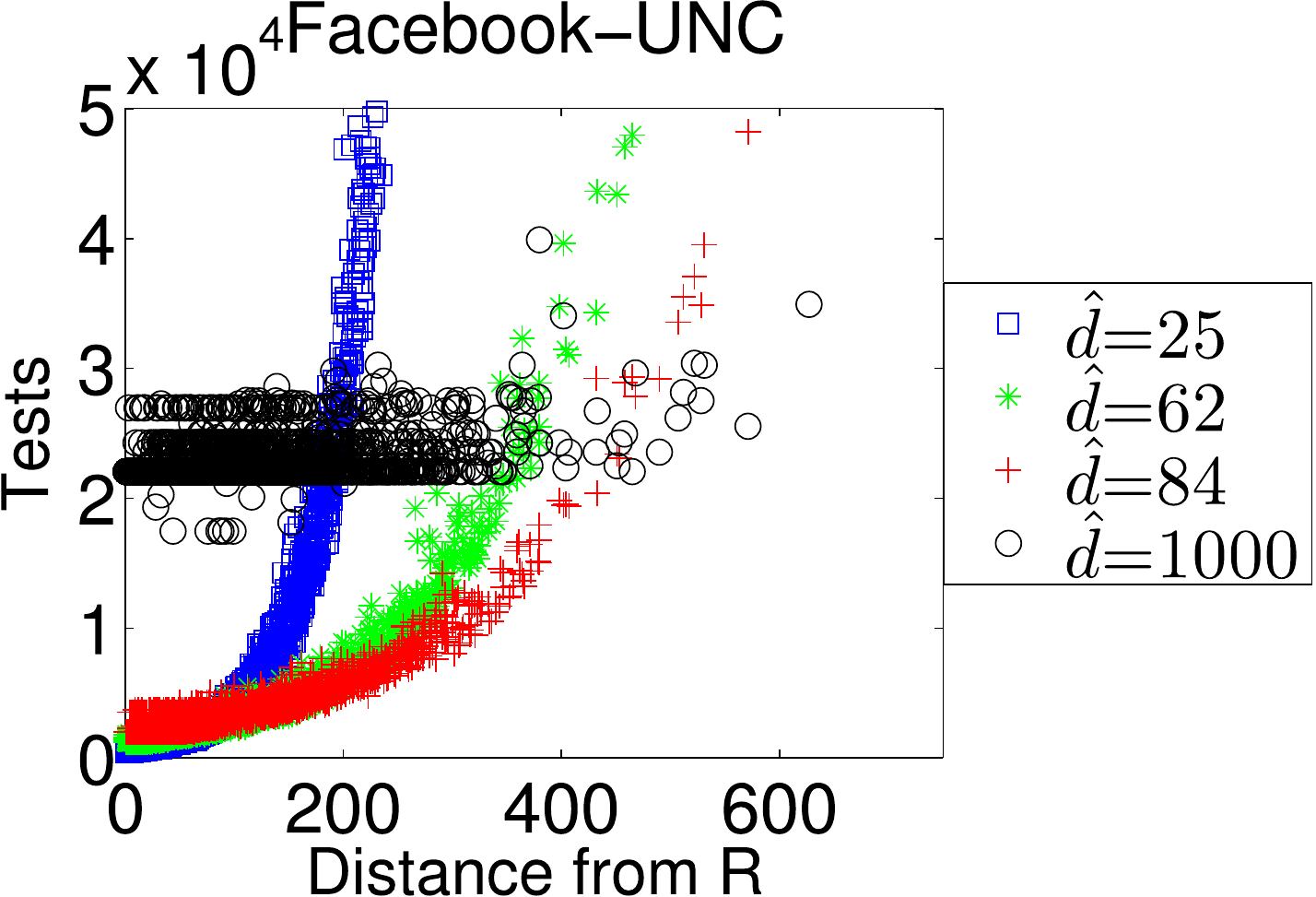} 
\caption{Number of tests required to clone each string, ordered by the 
string's distance from R.  Each string is represented by a dot. While the number of tests increases
rapidly for small $\hat{d}$ when the vector is far from $R$, note that many vectors are close to $R$, allowing for a majority of the database to be captured quickly.}
\label{fig:dot}
\end{figure*}

For each data set (excluding Slashdot-All and Facebook-Uniform due to their large scale), Figure \ref{fig:target} shows the number of tests needed to exactly clone a string (averaged across all strings in the database), as a function of $\hat{d}$.  
In a few instances, when the strings are very far from $R$, the algorithm may reach the 
cutoff value, causing the mean to be undervalued; thus, we also plot 
the median number of tests since the median is more robust to outliers.  Generally, we see that 
mean and median number of required tests decreases as $\hat{d}$ is decreased from $d$.  For 
instance, for the Slashdot database, the mean/median number of tests is 18,000 
if $\hat{d} = d = 378$, but if $\hat{d} = 50$, the mean/median number of tests is 3,000 and if
$\hat{d} = d_{\text{mean}} = 13$, the median number of tests requred is 700.  Sometimes, the 
mean number of tests increases if $\hat{d}$ is too small though.  If 
$\hat{d} = d_{\text{mean}} = 13$, the mean number of tests required is around 4,000.  Thus, there is 
a tradeoff.  If $\hat{d}$ is too small, it would take longer to exactly clone a string that is 
far away from $R$.  If $\hat{d}$ is too large (e.g. $\hat{d}=d$), then many inefficient tests would 
be performed on strings that are close to $R$.   
We assume that a good estimate for $\hat{d}$ (such as the median 
distance from $R$) can be obtained a priori, e.g. through scientific knowledge in the case of 
the Genomic database, or publicly available information in the cases of Netflix, Epinions, 
Slashdot, and Facebook.

We also investigate the
relationship between the number of required tests and the vector's
distance from $R$.  In Figure \ref{fig:dot}, we observe that the number of 
tests required to clone a vector is very low (and nearly constant) 
when the vector's distance from R is itself low and close to $\hat{d}$.  As the vector's distance increases, the number 
of required tests grows more quickly due to the mismatch between the 
distance and $\hat{d}$.  For each data set,
we display different scatter plots for different settings of $\hat{d}$.  For instance,
for the Slashdot data, the number of tests is relatively constant across all distances
when the $\hat{d}=200$; however, at this setting, the number of required tests is at least 10,000,
even when the vector is close to the reference $R$.  In contrast, when $\hat{d}=3$, the number of required tests
is only in the hundreds, around the vicinity of $\hat{d}$; however, when the vector's distance from $R$ is significantly
greater (e.g. over 100), the scatter plot increases dramatically.  It is important to note that most vectors are
close to $R$ due to the sparsity of the data, and thus, even when the scatter plot dramatically increases
when the distance from $R$ is great, there are relatively few vectors that fall within this regime.

\begin{figure}
\centering
\includegraphics[scale=0.3]{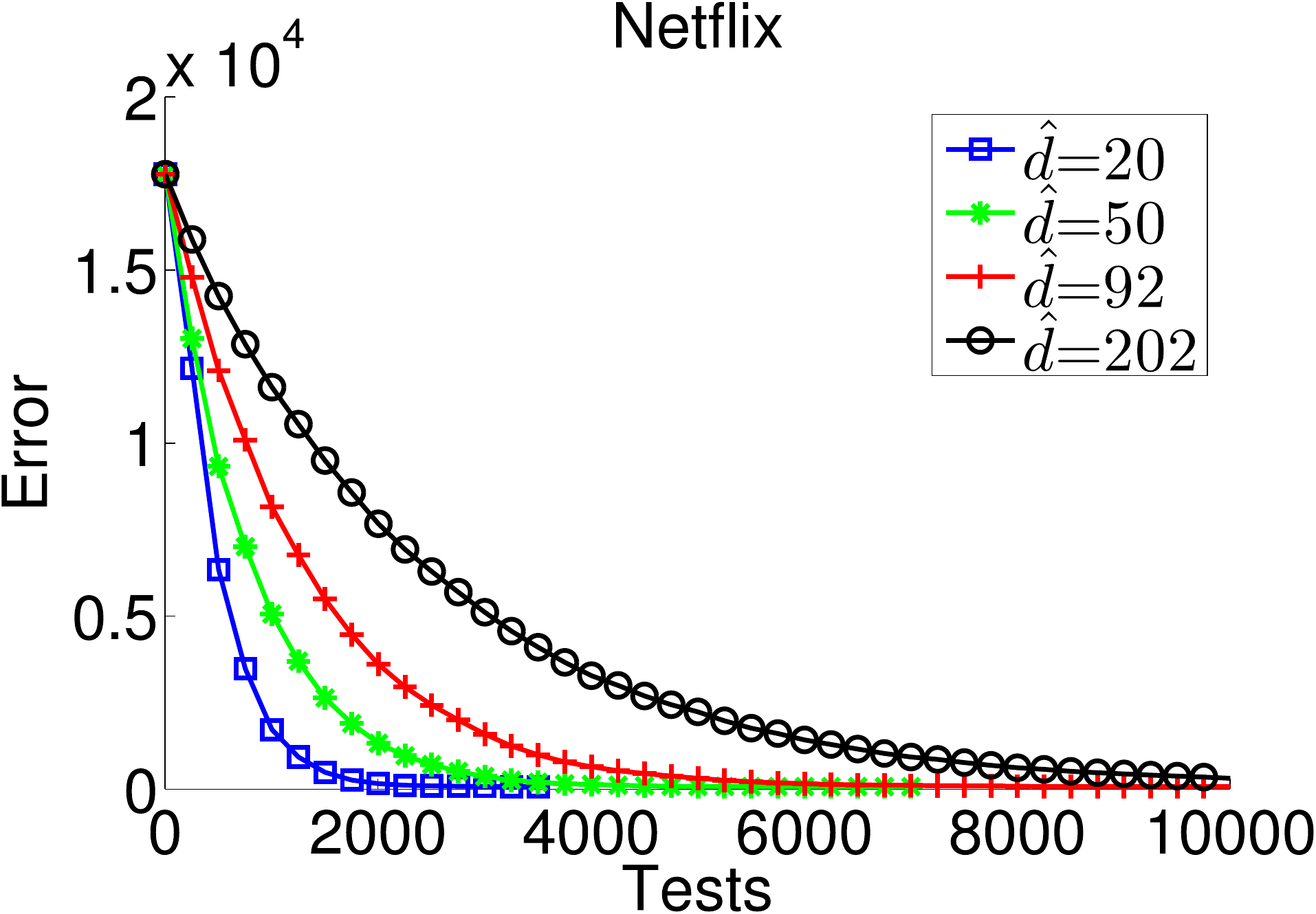}
\caption{Error as a function of the number of tests for a single Netflix user
who has rated 68 movies, for various $\hat{d}$.}
\label{fig:convergence}
\end{figure}

Providing another perspective, Figure \ref{fig:convergence} shows the decrease in error (defined as the number of 
differences between the string and the state of the reconstructed string) as the number of 
tests increases, for a randomly selected Netflix user who has rated 68 movies.  One can 
see that using $\hat{d}=202$ induces a slower rate of convergence than when using smaller 
settings for $\hat{d}$.  The case where $\hat{d}=d=1988$ is not shown since its rate of 
convergence is even slower.

\begin{figure}
\centering
\quad\includegraphics[scale=0.3]{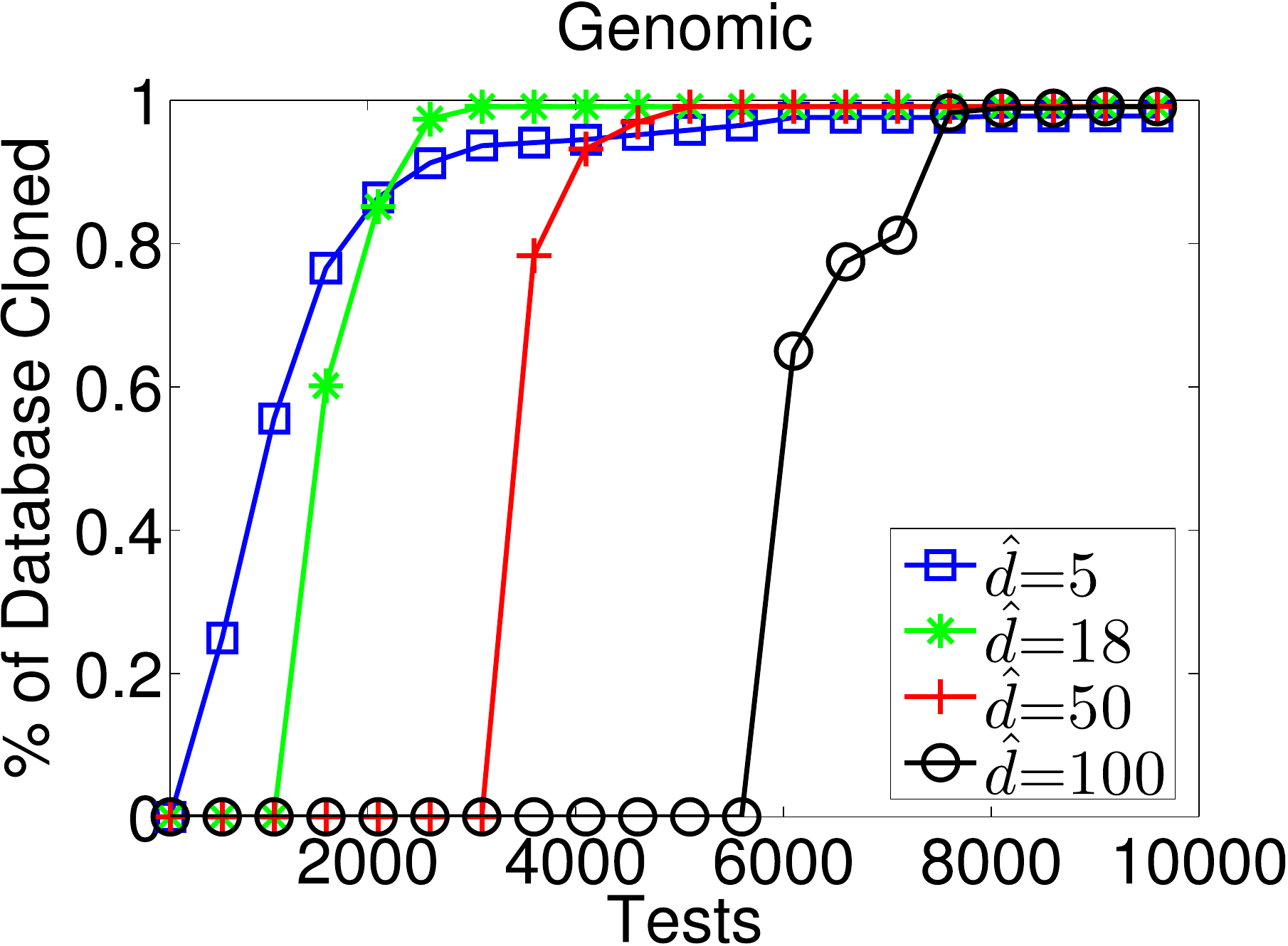} \qquad
\includegraphics[scale=0.3]{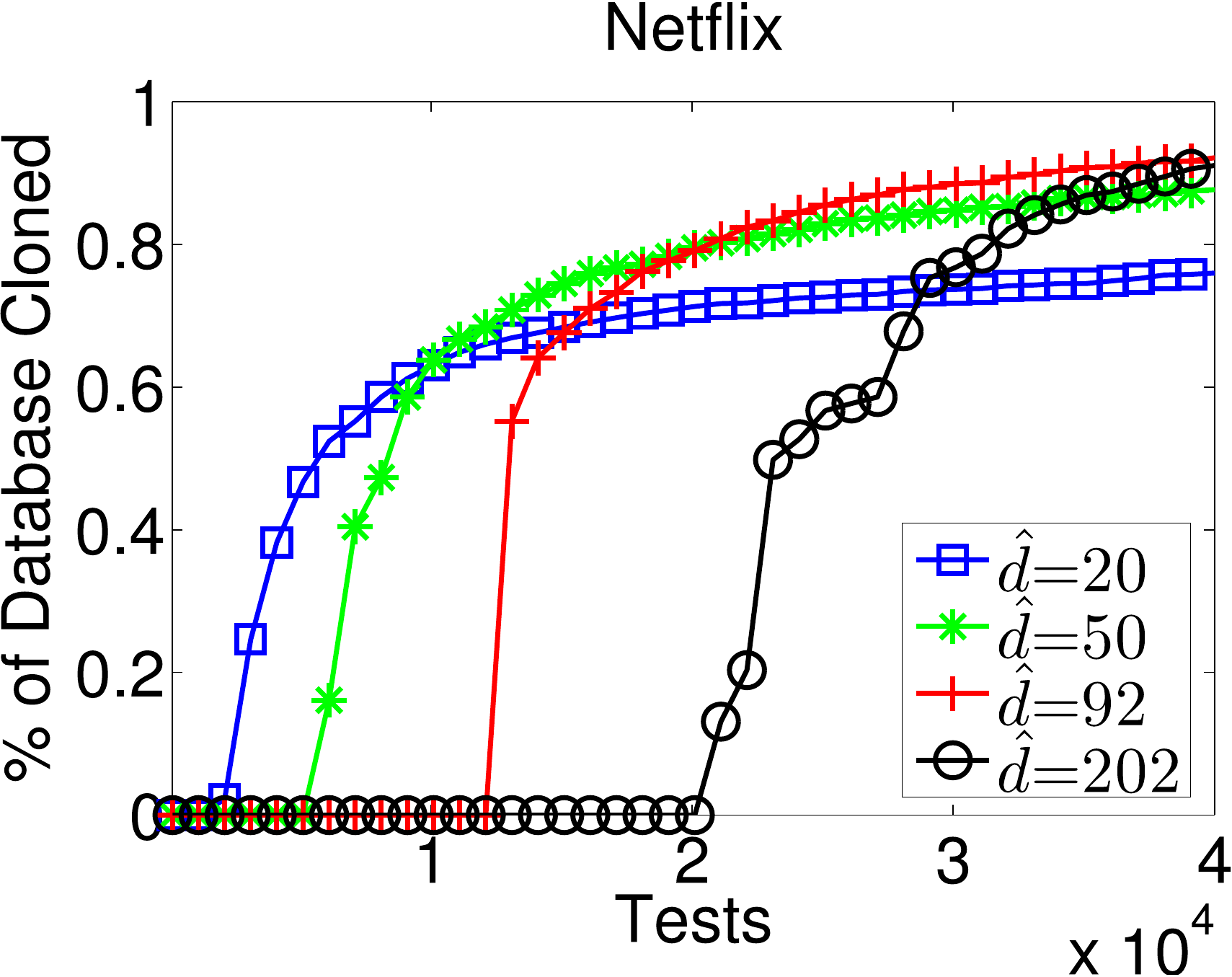} \\\quad \\
\includegraphics[scale=0.3]{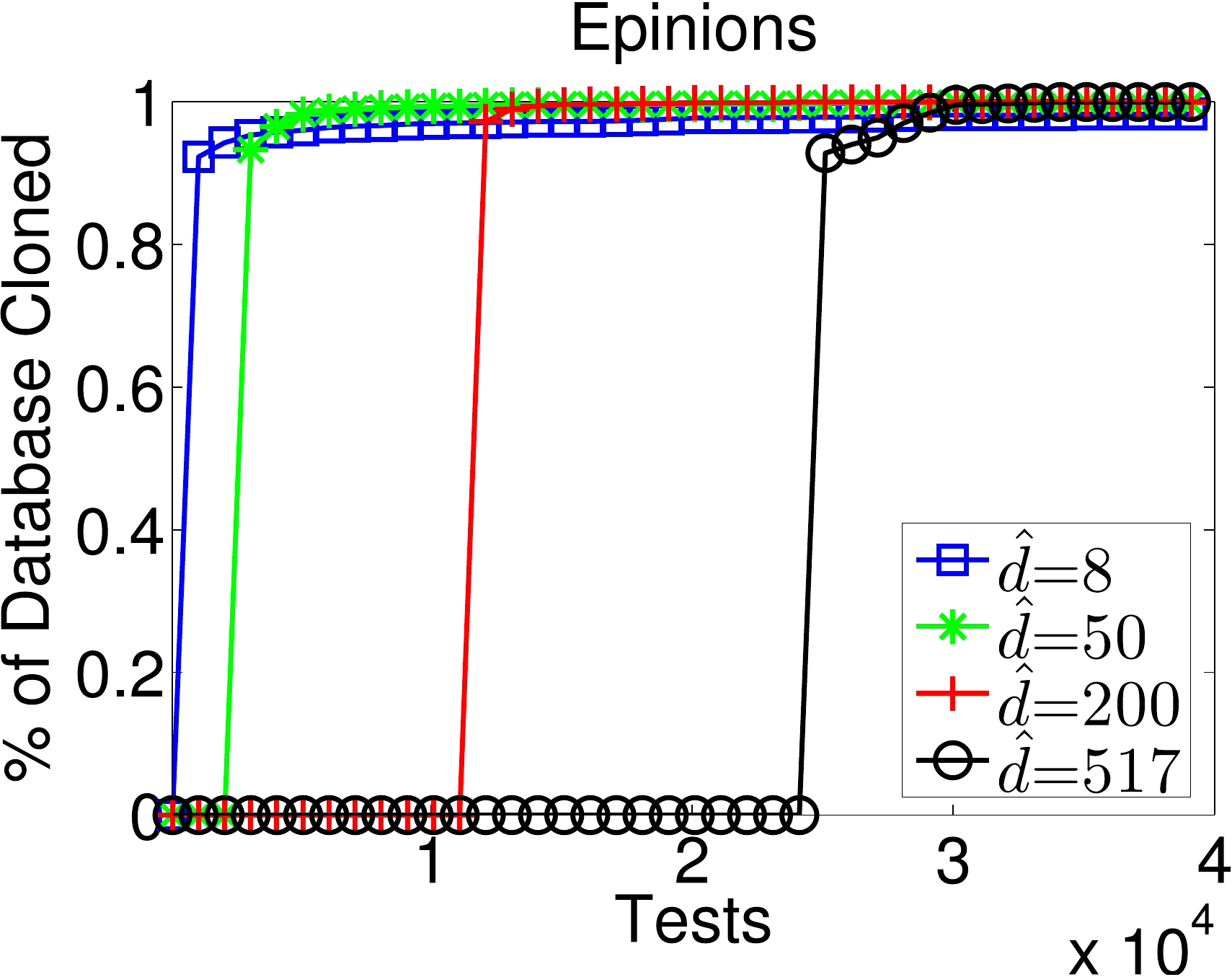} \qquad
\includegraphics[scale=0.3]{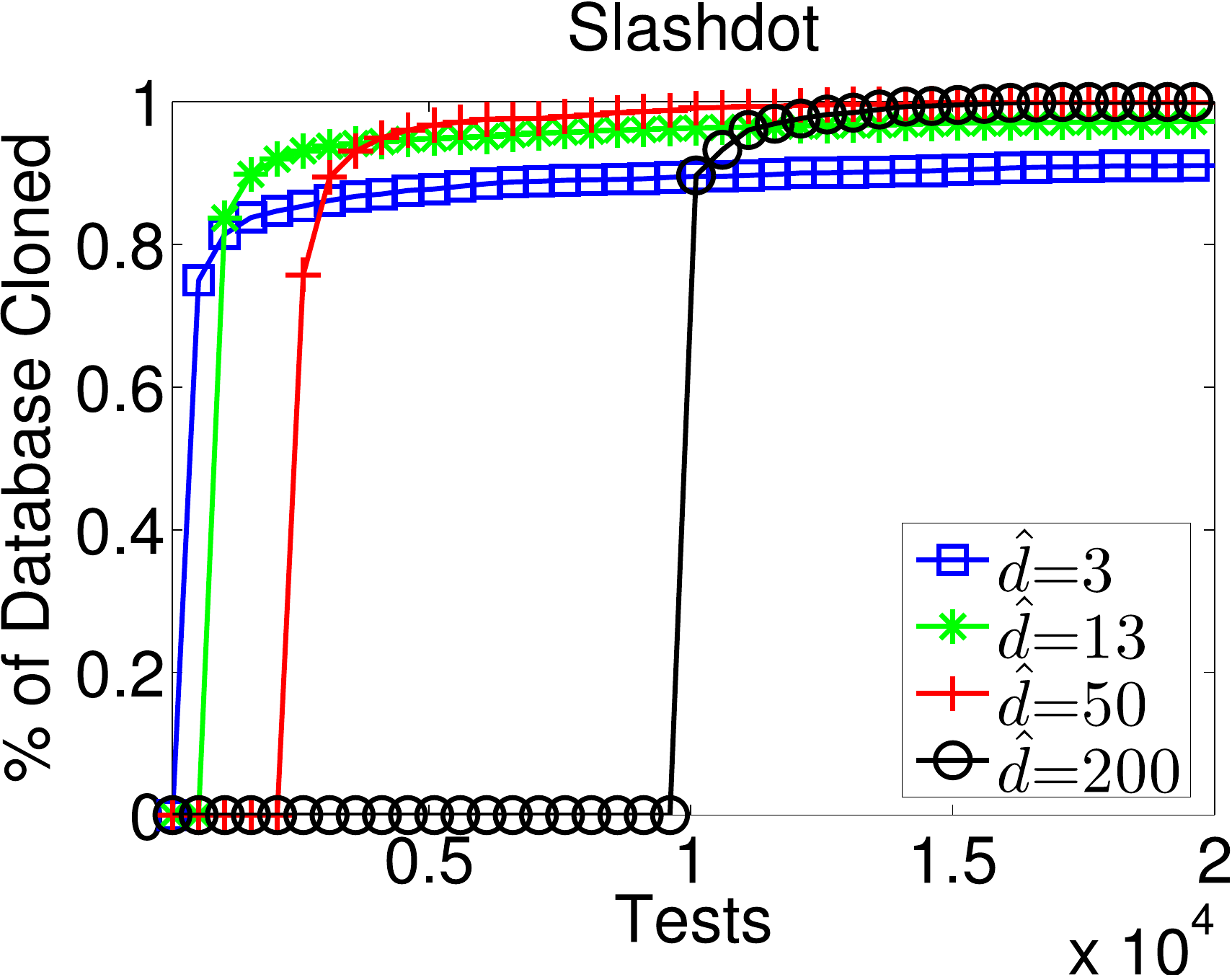} \\\quad \\
\quad\includegraphics[scale=0.3]{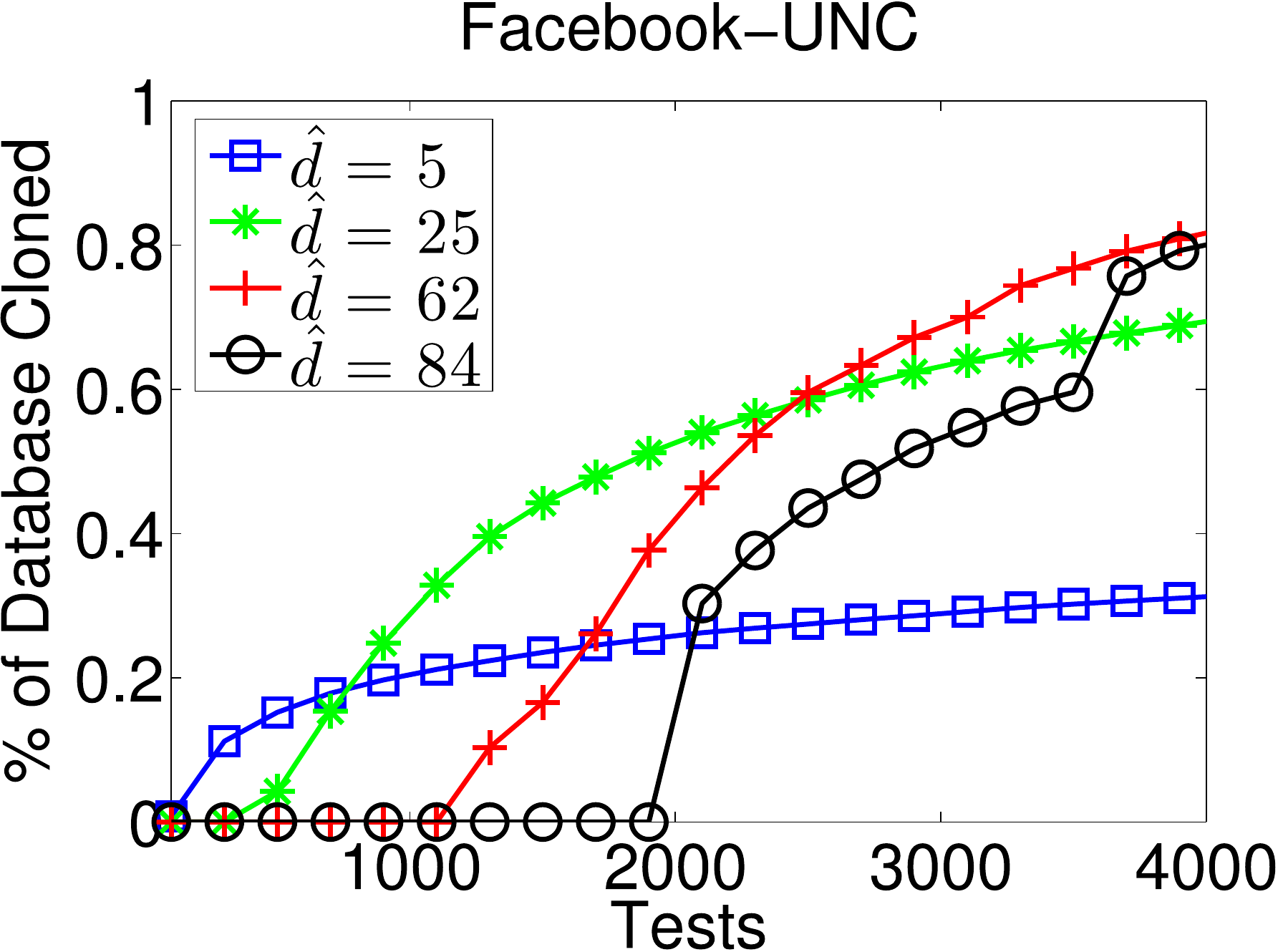} 
\caption{Percentage of strings cloned as a function of the number of tests, for each data set, using various $\hat{d}$.}
\label{fig:percentage}
\end{figure}

\begin{figure}
\centering
\includegraphics[scale=0.3]{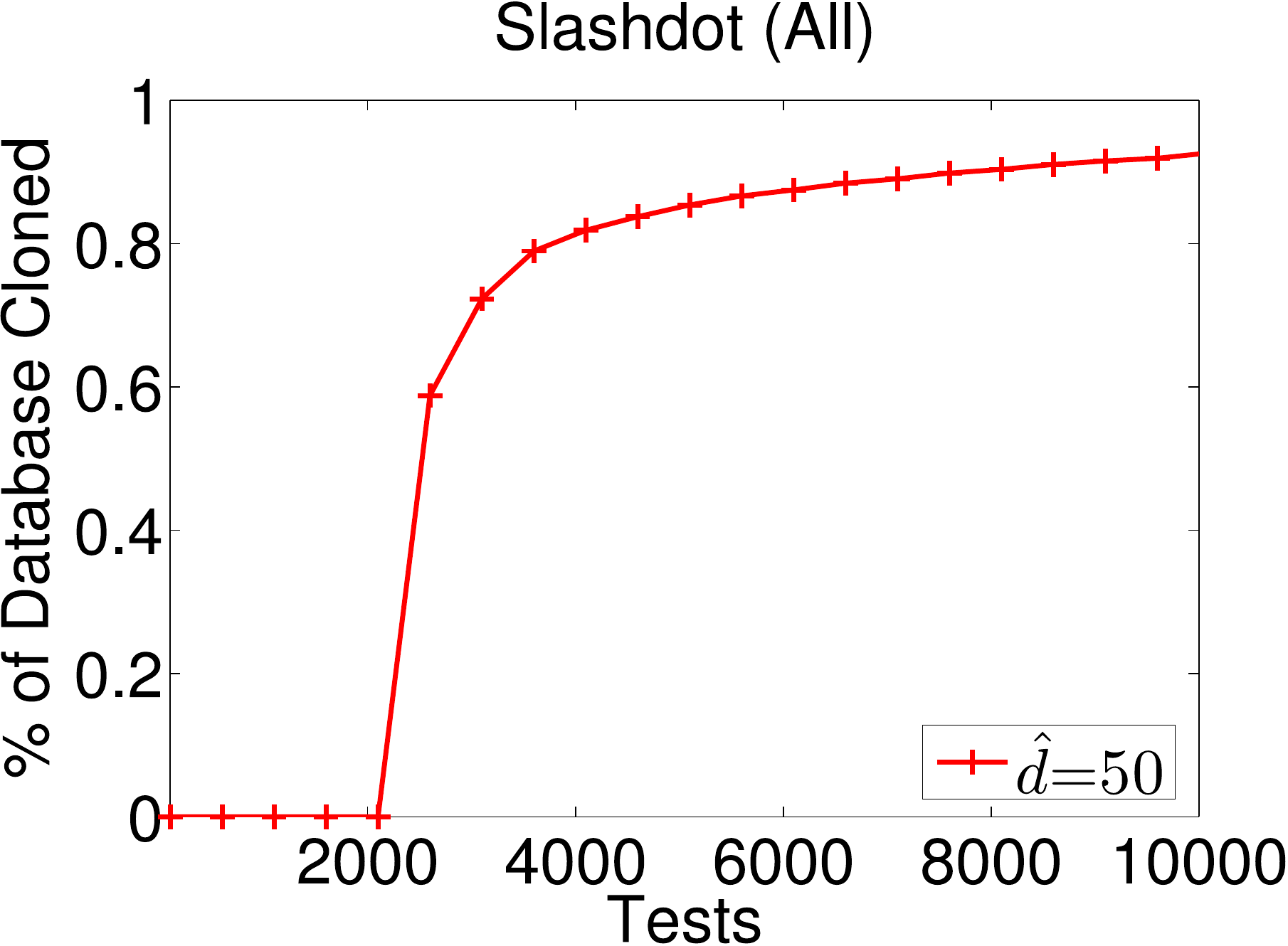} \qquad
\includegraphics[scale=0.29]{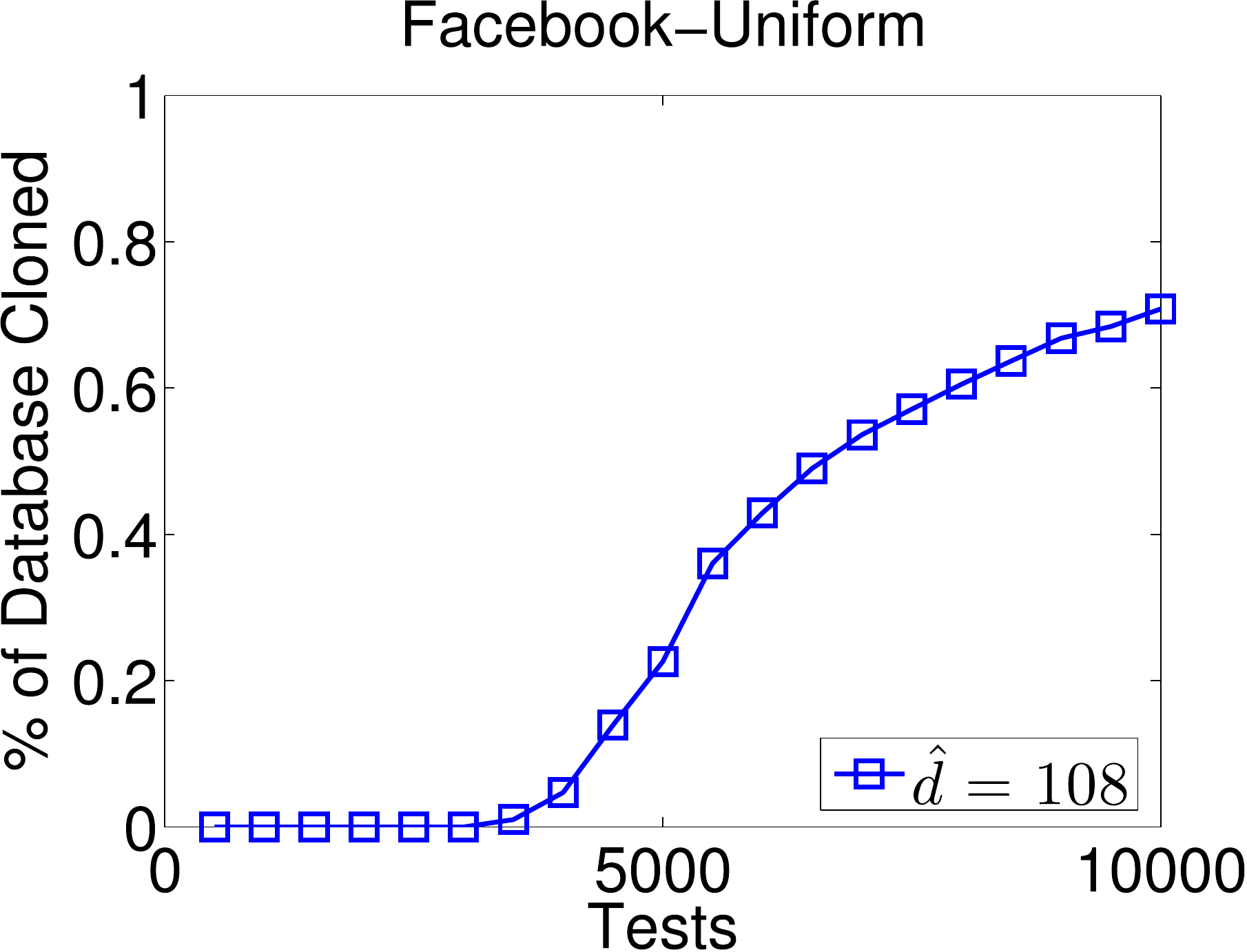}
\caption{Percentage of strings cloned as a function of the number of tests, for the large-scale
data sets.  Slashdot-All has a large number of strings
($g=82,144$) while Facebook-Uniform has large vector length ($n=72,261,577$). }
\label{fig:percentageLargeScale}
\end{figure}

In Figure \ref{fig:percentage}, the percentage of the database cloned by the
nonadaptive Mastermind attack is plotted as a function of the number of tests, 
for various $\hat{d}$.  We highlight some examples which demonstrate the 
efficiency of this attack.  For the Genomic data (using 
$\hat{d} = d_{\text{median}} = 18$), 78\% of the database is successfully recovered 
after 2,000 tests, and over 99\% of the database is recovered after 3,000 tests, which
is significantly less than both the baseline result and the theoretical bound in 
Table~\ref{tbl:theoretical}.
For the Netflix data (using $\hat{d} = d_{\text{median}} = 92$), 63\% of the strings 
are recovered after 10,000 tests.  For the Epinions data (using $\hat{d} = d_{\text{mean}} = 8$), 
68\% of the strings are recovered after only 500 tests.  For the Slashdot data 
(using $\hat{d} = d_{\text{mean}} = 13$) 82\% of the strings are recovered after 
only 1,000 tests.

For Facebook-UNC, we see that the Mastermind 
attack displays different behavior for different choices of $\hat{d}$.  When $\hat{d}=3$, the attack is able to 
quickly recover (the sparsest) 15\% of the data set after only 500 tests, but as the number of tests increases, 
the rate of progress slows significantly.  When $\hat{d}=25$, 52\% of the database has been successfully recovered 
after 2000 tests.  Thus, using only a couple thousand nonadaptive tests, we are able to 
clone the friend lists of half (9K out of 18K) of the Facebook users at the University of 
North Carolina.  

We also performed a large-scale nonadaptive Mastermind attack on Slashdot-All with 82,144 users.  
Figure \ref{fig:percentageLargeScale} shows that 
55\% of the strings are recovered after 2,500 
tests and that 81\% of the strings are recovered after 4,000 tests, using $\hat{d} = 50$.
Even when using a $\hat{d}$ which may be suboptimal, our empirical results suggest that it 
is possible to substantially outperform both the baseline method as well the theoretical 
bounds in Table \ref{tbl:theoretical} in practice, as long as $\hat{d}$ is chosen to be less than $d$.

We also ran the same experiment on Facebook-Uniform for $\hat{d}=108$ (the median distance from $R$). Figure~\ref{fig:percentageLargeScale} shows that over 70\% of the data set can be reconstructed with 10,000 tests, despite the fact
that the vector length of this data set is huge ($n=2,261,577$).   Since Facebook-Uniform contains an unbiased sample of 
users, these users are representative of the global Facebook population.  Furthermore, our theory states that
the number of required tests increases at a rate of at most $\log(g)$ where $g$ is the number of Facebook users.  In fact, the theoretical
number of tests needed to guarantee that 50\% of a 300-million user Facebook network is cloned is less
than 20,000 (assuming
$d_{\text{median}}=130$)\footnote{According to \url{http://www.facebook.com/press/info.php?statistics}, $d_{\text{mean}}=130$, and so 
$d_{\text{median}}$ should be even smaller, suggesting that the Mastermind attack can be even more efficient.}.  These results imply that an attacker may be able to recover 
over half of the global Facebook social network with several thousands of seemingly
innocuous nonadaptive Mastermind queries.  

It is worth noting that experiments have also been conducted on a variety of other data sets not
mentioned in this paper -- the nonadaptive Mastermind attack also performs 
very well on those data sets.  Results on cloning databases of binary attribute 
vectors (i.e. where the number of colors $c=2$) are described in previous work~\cite{AsuncionGoodrich2010}.

Our empirical results have shown that there is sensitivity to the choice of $\hat{d}$ in certain cases. 
One possible improvement 
is to use a tiered approach, where $\hat{d_1}$ is used to construct the first 5000 tests, 
$\hat{d_2}$ is used to construct the next 5000 tests, etc.   Each $\hat{d_i}$ could correspond 
to different mode.  Nonetheless, even when using a single $\hat{d}$, our results 
demonstrate that it is possible to clone a large fraction of a sparse 
database, by simply performing a nonadaptive Mastermind attack.

\section{Conclusion}
We have studied the Mastermind cloning attack, both from a
theoretical and experimental perspective, and have shown its
effectiveness in being able to copy the contents of a string database
through a sublinear number of string-comparison queries.  Furthermore, our
approach benefits from being fully nonadaptive and surreptitious in nature 
(due to the randomized
query construction), which is useful in real-world settings.

A natural direction for future work, of course, is on methods for
defeating our nonadaptive Mastermind attack, which we have not addressed in this paper.
Certainly, having Alice randomly permute the responses from her database 
with each query could help, since it would make it harder (but not
necessarily impossible) for Bob to correlate responses between
different queries. Of course, requiring Alice to always randomly permute
her responses would take extra time, and it may also require
additional space if she needs to store every response query so that
users can refer back to her responses for other, limited types of selection
queries she may allow.
So the technique of using random permutations can reduce the
risks associated with the Mastermind cloning attack, but it doesn't
necessarily eliminate these risks, and it comes with additional costs.

\subsection*{Acknowledgements}
A shorter version of the material in this paper (which only dealt with binary 
attribute vectors) was presented by the authors at
the ACM Workshop on Privacy in the Electronic Society (WPES), Chicago, October 2010.
We would like to thank Pierre Baldi and Padhraic Smyth for respectively 
suggesting the privacy of genomic data and Facebook relationships as research topics.
We are also grateful to Athina Markopoulou and Minas Gjokas for providing the Facebook-Uniform data.
We would also like to acknowledge David Eppstein and Daniel Hirschberg
for helpful discussions regarding the group-testing topics of this paper.
This research was supported by Office of Naval Research under MURI grant N00014-08-1-1015 and
by the National Science Foundation under grants 0724806, 0713046, 0847968, and an NSF Graduate Fellowship.

\bibliographystyle{plain}
\bibliography{goodrich,baldi,info,ref,genome,math,security,sigproc,k_anonymity,clone,cgt,experiments,social,related,extra,data}

\end{document}